\documentclass[journal,12pt,onecolumn,draftclsnofoot]{IEEEtran}
\usepackage{cite}
\usepackage{graphicx}  
\usepackage{tabulary}
\usepackage{amsfonts} 
\usepackage{amsmath}
\usepackage{chngpage}
\usepackage{makecell}
\usepackage{array}
\usepackage{amsthm}
\usepackage{booktabs}
\usepackage{amssymb}
\usepackage{multirow}
\usepackage{enumerate}
\usepackage{color}
\usepackage[dvipsnames]{xcolor}
\usepackage{mathrsfs}

\usepackage{float}
\usepackage{algorithm}  
\usepackage{algpseudocode}  
\usepackage{amsmath}  
\usepackage{url}

\usepackage{graphicx}
\usepackage{subfigure}
\usepackage[font=small]{caption}

\usepackage{mathtools}

\DeclareMathSizes{12}{12}{7.7}{5.5}
\newtheoremstyle{noparens}%
{}{}%
{\itshape}{}%
{\bfseries}{.}%
{ }%
{\thmname{#1}\thmnumber{ #2}\mdseries\thmnote{ #3}}

\theoremstyle{noparens}

\newtheorem{assumption}{Assumption}

\newtheorem{symmetry assumption}{Symmetry Assumption}
\newtheorem{lemma}{Lemma}

\newtheorem{remark}{Remark}
\newtheorem{proposition}{Proposition}
\newtheorem{definition}{Definition}

\begin{document}
	
	\title{Federated Learning with Lossy Distributed Source Coding: Analysis and Optimization}
	
	\author{
		\IEEEauthorblockN{Huiyuan Yang, Tian Ding, Xiaojun Yuan,~\IEEEmembership{Senior Member,~IEEE}}
	\thanks{This work has been submitted to the IEEE for possible publication. Copyright may be transferred without notice, after which this version may no longer be accessible.}
	}
	
	

	\maketitle
	\begin{abstract}

		Recently, federated learning (FL), which replaces data sharing with model sharing, has emerged as an efficient and privacy-friendly machine learning (ML) paradigm. One of the main challenges in FL is the huge communication cost for model aggregation. Many compression/quantization schemes have been proposed to reduce the communication cost for model aggregation. However, the following question remains unanswered: What is the fundamental trade-off between the communication cost and the FL convergence performance? In this paper, we manage to answer this question. Specifically, we first put forth a general framework for model aggregation performance analysis based on the rate-distortion theory. Under the proposed analysis framework, we derive an inner bound of the rate-distortion region of model aggregation. We then conduct an FL convergence analysis to connect the aggregation distortion and the FL convergence performance. We formulate an aggregation distortion minimization problem to improve the FL convergence performance. Two algorithms are developed to solve the above problem. Numerical results on aggregation distortion, convergence performance, and communication cost demonstrate that the baseline model aggregation schemes still have great potential for further improvement.

	\end{abstract}
	

	%
	\IEEEpeerreviewmaketitle

	\section{Introduction}
	
	Currently, there are nearly 6.3 billion smartphones \cite{O2021Number} and more than 11.3 billion connected Internet of Things (IoT) devices \cite{Sinha2021State} worldwide, which constantly collect/generate a wealth of data, such as videos, images, and measurements. In the conventional cloud-centric machine learning (ML) paradigm, all the training data is uploaded to a cloud server to produce effective inference models \cite{Lim2020Federated}. However, this centralized paradigm becomes unsatisfactory due to (i) the increasing sensitivity to data privacy and (ii) the increasing burden on the backbone network caused by the ever-growing data to be transmitted \cite{Imteaj2021A}. A decentralized ML paradigm called federated learning (FL) has been proposed to tackle these challenges. In FL, a number of user devices collaboratively train a global machine learning model with the help of a parameter server (PS). In each training iteration, the PS first broadcasts the parameters of the global model to some selected devices. Each selected device then computes a \emph{local update}, e.g., a gradient vector based on its local dataset, and then transmit it to the PS. Subsequently, the PS aggregates the locally computed updates to acquire a \emph{global update}. Finally, FL updates the parameters of the global model and proceeds to the next iteration.
	In this way, FL avoids direct data transmission and only requires exchanges of model parameters/updates, thereby reducing the overall communication cost without sacrificing data privacy. However, in each iteration, the selected user devices need to transmit their local updates to the PS, still entailing a significant volume of uplink transmission. Currently, the uplink communication cost appears to be a critical bottleneck in the employment of FL systems, especially for FL over wireless networks \cite{Liu2020Federated}.

	
	An interesting line of research to reduce the communication cost for FL over wireless networks is introducing the over-the-air computation (AirComp) technique into the FL uplink, referred to as over-the-air FL. In over-the-air FL, all selected user devices concurrently transmit their local updates using the same radio resource. By utilizing the signal superposition property of the multiple-access channel, AirComp has been shown to significantly relieve the communication bottleneck of FL \cite{Yang2020Federated, Liu2020Reconfigurable}. Nevertheless, over-the-air FL has some intractable deficiencies: (i) not directly deployable on current digital communication systems due to analog modulation; (ii) vulnerable to the stragglers \cite{Liu2020Reconfigurable}; (iii) difficult to combat the Byzantine attack effectively due to the uncoded nature \cite{Huang2021Byzantine}. Therefore, orthogonal FL uplink, where user devices are allocated with orthogonal resource units, is considered a more mature and practical setting by far.
	
	There is also a growing body of research aiming to design communication-efficient FL systems with orthogonal uplink. For example, the authors of \cite{McMahan2017Communication} proposed to select only part of the user devices to transmit their local updates. To save uplink communication resource, the selection is based on certain criteria such as link quality \cite{Luo2021Cost}. 
	Some works also proposed to exploit the sparsity in local gradients \cite{Aji2017Sparse, Lin2017Deep, Amiri2020Machine}. It was assumed that a portion of elements in local updates have very small magnitudes. These elements are considered negligible to the global model training, hence not being transmitted to the PS. The above schemes discard either a portion of local updates or a portion of elements in local updates. Such a coarse-grained discarding strategy could easily leave out exploitable information, potentially leading to a deterioration of the learning performance.
	
	
	Another popular approach to reduce the FL uplink cost is to adopt techniques of compression and/or quantization \cite{Bernstein2018signSGD, Alistarh2017QSGD, Wen2017Terngrad, Reisizadeh2020Fedpaq, Shlezinger2021UVeQFed}. For example, the authors in \cite{Bernstein2018signSGD} suggested to only transmit the signs of elements in the local updates, so as to reduce the payload of the uplink transmission. In \cite{Alistarh2017QSGD, Wen2017Terngrad, Reisizadeh2020Fedpaq}, various random scalar quantization methods are used to compress the local updates. The authors in \cite{Shlezinger2021UVeQFed} further proposed a lattice-based vector quantization scheme.
	
	
	In the above works for orthogonal FL uplink, the local updates from different user devices are treated as samples from independent information sources. The local updates are separately transmitted to the PS via orthogonal channels. Based on the received signals, the PS first decodes all the local updates independently, and then aggregates them to generate the global model update. However, it was observed that the local updates are not independent in practical machine learning tasks but possess a significant correlation among user devices \cite{Chen2020Scalecom, Zhong2021Over}. This correlation, if properly utilized, can potentially reduce a great amount of communication cost for model aggregation. Furthermore, the PS does not have to estimate all the local updates but only the global update, which is a function of the local updates. Naturally, we have the following question: What is the performance limit of federated learning, especially when the adopted model aggregation scheme takes full advantage of the above two properties?
	

	In this paper, we manage to answer the above question. The main contributions are listed as follows:
	\begin{itemize}
		\item We put forth a general information-theoretic analysis framework for the analysis of the model aggregation performance. In this analysis framework, the encoding, transmission, and aggregation (decoding) of the local updates are unified as a lossy distributed source coding (DSC) problem \cite{Oohama2005Rate, Wagner2008Rate, Krithivasan2009Lattices, Wagner2011On, El2011Network}.
		
		\item Under the proposed analysis framework, we derive an inner bound of the rate-distortion region of model aggregation by giving an achievability scheme.
		
		\item We conduct an FL convergence analysis to characterize the relationship between the FL convergence performance and the aggregate distortion. We further develop two algorithms (for general and symmetric FL systems, respectively) to search for the point with minimum aggregation distortion in our proposed inner region.
	\end{itemize}
	Numerical results are provided to evaluate the performance gap between baseline model aggregation schemes and our theoretical bound in terms of aggregation distortion, convergence performance, and communication cost. The results demonstrate that the baseline model aggregation schemes still have great potential for further improvement in the considered scenarios.

	The remainder of this paper is organized as follows. In Section \ref{Lossy-DSC-Framework-for-Federated-Learning}, we introduce the FL system and formulate a framework for aggregation performance analysis. In Section \ref{modified-Berger-Tung-coding}, we derive an inner bound of the rate-distortion region of model aggregation. Sections \ref{Coding-Parameters-Design} and \ref{MBTC-Parameters-Optimization-with-Symmetric-Assumptions} develop two algorithms to minimize the aggregation distortion for general and symmetric FL systems, respectively. In Section \ref{Numerical-Results}, we present the numerical results. Finally, conclusions are drawn in Section \ref{Conclusion}.


	%

	$Notation:$ Scalars, vectors, and matrices are denoted by regular letters (lower-case or upper-case), bold lower-case letters and bold upper-case letters, respectively. The transpose of a vector $\mathbf{a}$ or a matrix $\mathbf{A}$ is denoted by $\mathbf{a}^\top$ and $\mathbf{A}^\top$, respectively. $\mathbf{0}$, $\mathbf{1}$ denote all-zero or all-one vectors or matrices, respectively, and $\mathbf{I}$ denotes the identity matrix. We use $[\mathbf{a}]_m$ to represent the $m$-th element in vector $\mathbf{a}$. We use $\mathbb{R}^d$, $\mathbb{Z}_+$ to represent the $d$-dimensional Euclidean space and the positive integer set, respectively. We also use $[n]$ as the abbreviation of $\{1,2,\cdots,n\}$. Given a set $\mathcal{S} \!\subseteq\! [d]$, $\mathcal{S}^c$ denotes the set $[d] \backslash \mathcal{S}$. Given a vector $\mathbf{a} \!\in \! \mathbb{R}^{d}$ and a nonempty set $\mathcal{S} \!\subseteq\! [d]$, $\mathbf{a}^{\mathcal{S}}$ denotes the $|\mathcal{S}|$-dimensional vector obtained by removing all the $i$-th elements of $\mathbf{a}$ with $i \!\in\! \mathcal{S}^c$. Similarly, given a matrix $\mathbf{A} \!\in\! \mathbb{R}^{d_1 \times d_2}$ and nonempty sets $\mathcal{S}_1\!\subseteq\! [d_1]$, $\mathcal{S}_2 \!\subseteq\! [d_2]$, $\mathbf{A}^{\mathcal{S}_1, \mathcal{S}_2}$ denotes the $|\mathcal{S}_1|\!\times\! |\mathcal{S}_2|$ matrix obtained by removing all the ($i$, $j$)-th elements of $\mathbf{A}$ with $i \!\in\! [d_1]\backslash \mathcal{S}_1$ or $j \!\in\! [d_2] \backslash \mathcal{S}_2$. When $\mathcal{S}_1 \!=\! \mathcal{S}_2 \!=\! \mathcal{S}$, we will simplify the notation $\mathbf{A}^{\mathcal{S}_1, \mathcal{S}_2}$ as $\mathbf{A}^{\mathcal{S}}$. We use $\mathcal{N}(\mu, \sigma^2)$ to denote normal distribution with mean $\mu$ and variance $\sigma^2$, and $\mathcal{N}(\boldsymbol{\mu}, \boldsymbol{\Sigma})$ to denote the multivariate normal distribution with mean vector $\boldsymbol{\mu}$ and covariance matrix $\boldsymbol{\Sigma}$.
	
	Let $f$ be a scalar-value function with $n$ scalar inputs. We say that $f$ is applied element-wise on $n$ vectors $\mathbf{a}_1,\cdots, \mathbf{a}_n \in \mathbb{R}^d$, if $f$ outputs a vector $\mathbf{b} = f(\mathbf{a}_1,\cdots, \mathbf{a}_n) \in \mathbb{R}^d$ with each element given by $[\mathbf{b}]_i = f([\mathbf{a}_1]_i, \cdots, [\mathbf{a}_m]_i)$, $\forall i\in [d]$.

		\section{System Model and Aggregation Performance Analysis Framework}
	\label{Lossy-DSC-Framework-for-Federated-Learning}
	\subsection{Federated Learning System} \label{Federated Learning System}
	We consider a federated learning (FL) system comprising a central parameter server (PS) and $M$ distributed user devices. The objective of the FL system is to cooperatively train a global machine learning model (parameterized by vector $\boldsymbol{\theta}\in \mathbb{R}^N$) based on the data collected by the user devices. Specifically, in FL, each device $m$ is only allowed to access its local dataset $\mathcal{D}_m = \{\boldsymbol{\zeta}_{m, k}\}_{k=1}^{K_m}$, where $K_m$ is the sample size and $\boldsymbol{\zeta}_{m, k}$ is the $k$-th training sample collected by device $m$\footnote{For example, in supervised learning each training sample $\boldsymbol{\zeta}_{m, k}$ consists of a feature vector and a corresponding label.}. For each device $m$, we define a local empirical loss function with respect to the global model parameter $\boldsymbol{\theta}$, given by
	\begin{equation}
	\setlength{\abovedisplayskip}{3pt}
	\setlength{\belowdisplayskip}{3pt}
	L_m(\boldsymbol{\theta};\mathcal{D}_m) \triangleq \frac{1}{K_m} \sum_{k=1}^{K_m} l(\boldsymbol{\theta}; \boldsymbol{\zeta}_{m, k}),
	\end{equation}
	where $l(\boldsymbol{\theta}; \boldsymbol{\zeta}_{m, k})$ denotes the sample-wise loss function. FL aims to minimize the global empirical loss function, i.e.,
	\begin{equation}
	\setlength{\abovedisplayskip}{3pt}
	\setlength{\belowdisplayskip}{3pt}
	\min_{\boldsymbol{\theta}} L(\boldsymbol{\theta}) \triangleq \frac{1}{K} \sum_{m=1}^{M} K_m L_m(\boldsymbol{\theta};\mathcal{D}_m),
	\end{equation}
	where $K = \sum_{m=1}^M K_m$ denotes the total sample size.
	
	
	FL involves multiple rounds of iteration for convergence.
	At the $t$-th iteration round, FL  performs the following four steps:
	\begin{itemize}
		\label{FL_four_steps}
		\item [(i)] \emph{Model broadcast}: The PS broadcasts the current global model parameter $\boldsymbol{\theta}^{(t)}$ to all the devices.
		\item [(ii)] \emph{Local update computation}: Each device $m$ computes a local update $\mathbf{g}_{m}^{(t)} \in \mathbb{R}^N$ on the basis of the received $\boldsymbol{\theta}^{(t)}$ and the local dataset $\mathcal{D}_m$.
		
		\item[(iii)] \emph{Model Aggregation}:
		\begin{itemize}
			\item \emph{Encoding}: Each device $m$ properly encodes its local update $\mathbf{g}_{m}^{(t)}$ into codewords.
			\item \emph{Transmission}: Each device $m$ transmits its messages to the PS via a bit-constrained error-free link \cite{Kone2016Federated, Alistarh2017QSGD, Shlezinger2021UVeQFed}.\footnote{In this paper, we do not restrict the communication links to be wired or wireless, but only requires them to be independent/orthogonal without interfering with each other.}
			\item \emph{Decoding}: After the PS receives all the codewords, it performs joint decoding to obtain $\hat{\mathbf{g}}^{(t)}$, which is an estimation of global update $\mathbf{g}^{(t)}$.\footnote{Since the transmission links from the user devices to the PS are bit-constrained and thus can not transmit continuous-valued local updates losslessly, the PS ends up with only an estimation of $\mathbf{g}^{(t)}$.} The global update $\mathbf{g}^{(t)}$ is a function of the local updates, denoted by 
			\begin{equation}
			\label{eq::Model-aggregation}
			\mathbf{g}^{(t)} \triangleq
			\kappa\left(\mathbf{g}_1^{(t)}, \dots, \mathbf{g}_M^{(t)}\right),
			\end{equation}
			where $\kappa: \underbrace{\mathbb{R} \times \dots \times \mathbb{R}}_M \mapsto \mathbb{R}$ is the \emph{aggregation target function} and is applied element-wise to the $M$ local updates in \eqref{eq::Model-aggregation}.\footnote{A more general way is to define the aggregation target function as a vectors-to-vector mapping. However, this broader definition might cause unnecessary difficulty in understanding. Thus, we use an element-wise function here.}
		\end{itemize}
	
		
		\item[(iv)] \emph{Global model update}: The PS updates the global model by $\hat{\mathbf{g}}^{(t)}$ with a learning rate $\eta$, i.e., $\boldsymbol{\theta}^{(t+1)} = \boldsymbol{\theta}^{(t)} - \eta \hat{\mathbf{g}}^{(t)}$.
	\end{itemize}

	\subsection{Analysis Framework for Model Aggregation Performance}
	\label{sec::lossy-DSC-framework}
	In this subsection, we propose an analysis framework for model aggregation performance from a rate-distortion theory perspective \cite{Oohama2005Rate, Wagner2008Rate, Krithivasan2009Lattices, Wagner2011On, El2011Network}. 
	We consider the $M$ $N$-dimensional local updates at the $t$-th round to be $M$ vectors randomly generated from a joint distribution $\mathcal{P}^{(t)}_{N}$, i.e.,
	\begin{equation}
	\setlength{\abovedisplayskip}{3pt}
	\setlength{\belowdisplayskip}{3pt}
	\left(\mathbf{g}^{(t)}_1, \mathbf{g}^{(t)}_2, \cdots, \mathbf{g}^{(t)}_M\right) \sim \mathcal{P}^{(t)}_{N}.
	\end{equation}
	We emphasize that the elements in each $\mathbf{g}^{(t)}_m$ are not necessarily independently or identically distributed.
	In the following, we focus on an arbitrary iteration round $t$ and omit the superscript $(t)$ for brevity whenever causing no ambiguity. We further add a superscript of $N$ to clarify the model dimension, i.e., denoting the local update by $\mathbf{g}^N_m$, $\forall m\in[M]$.
	

	We now introduce the performance analysis framework for the model aggregation step. First, each user device employs an encoder \begin{equation}\label{encoder-mapping}
	\setlength{\abovedisplayskip}{3pt}
	\setlength{\belowdisplayskip}{3pt}
	f_m^N: \mathbb{R}^N ~\mapsto ~[B_m^N], ~~ \forall m \in [M],
	\end{equation}
	that maps its local update to a positive integer $b^N_m = f^N_{m}(\mathbf{g}^N_m)$ and transmits $b^N_m$ to the PS. Upon receiving the codewords $\{b^N_m\}_{m=1}^M$, the PS employs a joint decoder 
	\begin{equation}\label{decoder-mapping}
	\psi^N: [B_1^N] \times \dots \times [B_M^N]~ \mapsto~ \mathbb{R}^N,
	\end{equation}
	to generate $\hat{\mathbf{g}}^N = \psi^N(b_1^N, \cdots, b_m^N)$, which is an estimation of the global update $\mathbf{g}^N = \kappa(\mathbf{g}_1^N, \dots, \mathbf{g}_M^N)$. Let $d: \mathbb{R}^N \times \mathbb{R}^N \mapsto \mathbb{R}_+$ be a distortion measure between two $N$-dimensional vectors. Then, we have the following definition:
	\begin{definition}
		\label{achievalbe_definition}
		A rate-distortion tuple $(R_1, \dots, R_M, D)$ is said to be achievable if for any $\epsilon > 0$ and any sufficiently large $N$, there exists $M$ encoders $\{f_m^N\}_{m=1}^M$ and a joint decoder $\psi^N$ such that rate $\displaystyle{\frac{1}{N}}\log\left(B_m^N\right) \leq R_m + \epsilon$, $\forall m \in [M]$, and expected aggregation distortion $\mathbb{E}[d(\mathbf{g}^N, \hat{\mathbf{g}}^N)] \leq D + \epsilon$.
	\end{definition}
	Loosely speaking, if a rate-distortion tuple $(R_1, \cdots, R_M, D)$ is proved to be achievable, then, as the model dimension $N$ increases, the expected aggregation distortion $\mathbb{E}[d(\mathbf{g}^N, \hat{\mathbf{g}}^N)]$ can be less than and arbitrarily close to $D$ with the rate $\log\left(B_m^N\right)/N$ less than and arbitrarily close to $R_m$, $\forall m \in [M]$. We also call $R_m$ and $D$ rate and distortion, respectively. The \emph{rate-distortion region} of model aggregation is defined as the set of all achievable rate-distortion tuples, denoted by $\mathcal{RD}^{\star}$. In the next section, we analyze the model aggregation performance by characterizing an inner bound of $\mathcal{RD}^{\star}$ under the quadratic distortion measure and the linear aggregation target function.
	\begin{figure*}[htbp]
		\centering
		\includegraphics[width=13cm]{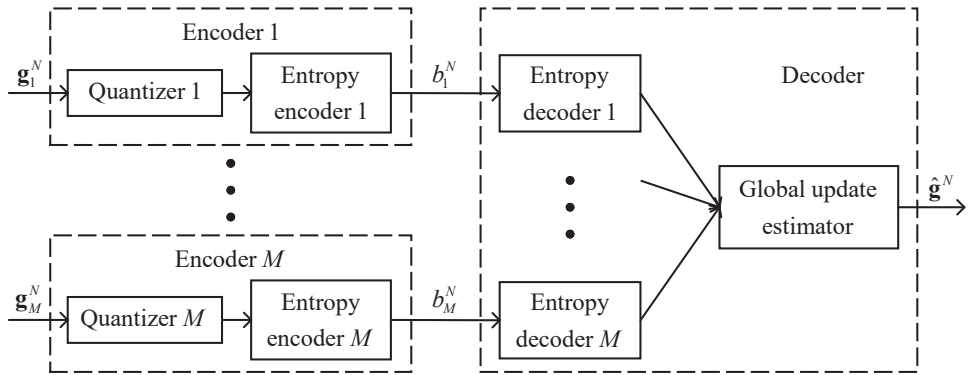}
		\caption{Conventional model aggregation schemes for orthogonal FL uplink.}
		\label{fig::uplink-baseline}
	\end{figure*}

	\begin{remark} 
	\label{remark-for-Fig1}
	Most of the model aggregation schemes in mainstream FL research \cite{Bernstein2018signSGD, Alistarh2017QSGD, Wen2017Terngrad, Reisizadeh2020Fedpaq, Shlezinger2021UVeQFed} are covered by our analysis framework in Fig. \ref{fig::uplink-baseline}.
	\end{remark}
	
	\begin{remark} 
		\label{literature_reviews}
		Our analysis framework for model aggregation performance is similar to the framework of the well-studied distributed function computation problem (please refer to \cite{Korner1979How, Krithivasan2009Lattices, Wagner2011On, Krithivasan2011Distributed, Lim2019Towards} and references therein for further details). The main difference between them is that our framework allows the existence of correlations among the elements of the local updates, which better matches the application of FL (where the elements of the local updates are generally correlated \cite{Xue2022FedOComp, Zhong2021Over}). More specifically, we essentially model the FL model aggregation problem as a lossy compression problem for sources with memory rather than for memoryless sources considered in the distributed function computation problem. In the next section, to give a constructive analysis, we introduce some randomness (in the random rotation step) to break these element-wise correlations.
		
			
	\end{remark}

	\section{An Inner Region of $\mathcal{RD}^{\star}$} \label{modified-Berger-Tung-coding}

\begin{figure*}[htbp]
	\vspace{-0.1cm}
	\setlength{\belowcaptionskip}{-0.3cm}
	\centering
	\includegraphics[width=13cm]{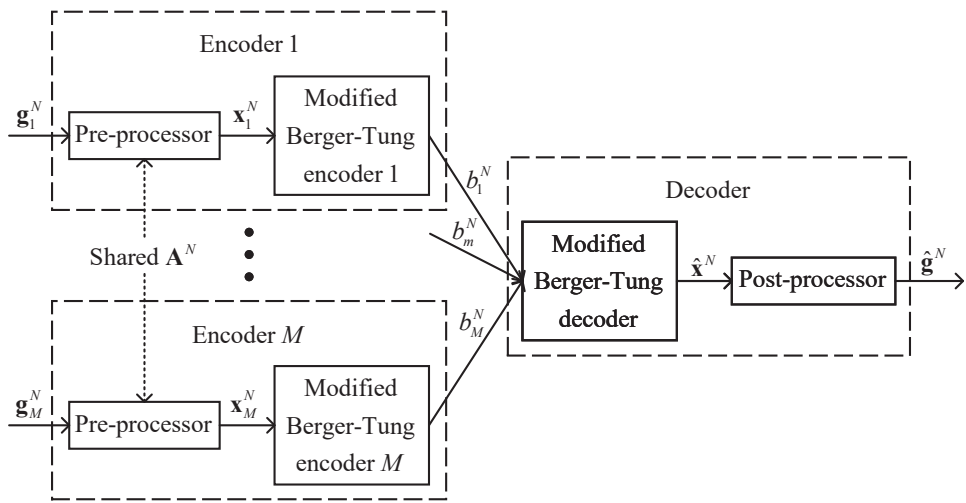}
	\caption{The achievability scheme characterizing the inner region $\mathcal{RD}_{\mathrm{in}}$.}
	\label{MA-BT}
\end{figure*}


In this section, we give an achievability scheme, which characterizes an inner region of $\mathcal{RD}^{\star}$, denoted by $\mathcal{RD}_{\mathrm{in}}$. Following the common practice \cite{Shlezinger2021UVeQFed, Chen2020Scalecom, Zhong2021Over, Abdi2019Reducing}, we consider a quadratic distortion measure $d(\mathbf{g}^N, \hat{\mathbf{g}}^N) = \|\mathbf{g}^N - \hat{\mathbf{g}}^N\|^2/N$ and a linear aggregation target function $\kappa(\mathbf{g}^N_1, \dots, \mathbf{g}^N_M) = \sum_{m=1}^{M} c_m \mathbf{g}^N_m$.

Fig. \ref{MA-BT} summarizes the achievability scheme. First, the local update of each device $m$ is fed into the $m$-th modified Berger-Tung encoder after pre-processing, and then each encoder encodes the preprocessed data into a codeword with rate $R_m$. The modified Berger-Tung decoder performs joint decoding after receiving all the codewords, whose output is post-processed to obtain an estimation of the global update. In the remainder of this section, We first introduce the specific operations performed by each module in Fig. \ref{MA-BT} and then characterize $\mathcal{RD}_{\mathrm{in}}$.

\subsection{Data Processing}
\label{Data-Processing}
\subsubsection{Data pre-processing}
The data pre-processing procedure consists of two steps: \emph{mean removal} and \emph{random rotation}. The main objective is to process the local updates so they can be modelled as samples from memoryless correlated Gaussian sources asymptotically.

In the mean removal step, Each device $m$ computes the average of the elements of $\mathbf{g}_m^{N}$ as $\bar{g}_m = \sum_{n=1}^N[\mathbf{g}_m^{N}]_n/N$, which are then uploaded to the PS.\footnote{We assume that the cost of transmitting the scalars $\{\bar{g}_m\}_{m=1}^M$ is negligible relative to that of transmitting the local updates.} Subsequently, each device $m$ computes the mean-removed vector
\begin{equation}
\label{mean_removal}
\tilde{\mathbf{g}}_m^{N} \triangleq \mathbf{g}_m^{N} - \bar{g}_m\mathbf{1}_N.
\end{equation}

In the random rotation 
step, the PS and all the devices generate a shared Haar distributed\footnote{That is, uniformly distributed on the set of orthogonal matrices.} matrix $\mathbf{A}^{N} \in \mathbb{R}^{N \times N}$ using public randomness. Each device $m$ computes
\begin{equation}\label{random-rotation}
\mathbf{x}_m^N = \mathbf{A}^{N} \tilde{\mathbf{g}}_m^{N} \in \mathbb{R}^{N}, \ \forall m \in [M],
\end{equation}
then feeds $\mathbf{x}^N_m$ into the modified Berger-Tung encoder.

In the following, we show that, under certain assumptions on $\{\mathbf{g}^N_m\}^M_{m=1}$, the resulting $\{\mathbf{x}_m^N\}_{m=1}^M$ can be asymptotically approximated by correlated Gaussian vectors in the sense of the quadratic distortion. We state the assumptions and the corresponding consequences as follows.
\begin{assumption} \label{linear-combination-assumption}
	(Correlation model) The $M$ sequences of the mean-removed vectors $\{\tilde{\mathbf{g}}_m^N \in \mathbb{R}^N \}_{N \in \mathbb{Z}_+}$, $m\in [M]$, can be modeled as $\tilde{\mathbf{g}}_m^{N} = \sum_{k=1}^K e_{m, k} \mathbf{p}_k^N$, where each $e_{m, k} \in \mathbb{R}$ is a constant coefficient, and the base vectors $\{\mathbf{p}_k^N \in \mathbb{R}^N\}_{N \in \mathbb{Z}_+}$, $k\in[K]$, are $K$ sequences of random vectors satisfying
	\begin{enumerate}[(i)]
		\item $\{\mathbf{p}_k^N\}_{k=1}^K$ are mutually independent, $\forall N \in \mathbb{Z}_+$;\label{Assumption1-ii-a}
		\item $\{\mathbf{p}_k^N\}_{k=2}^K$ are isotropically distributed\footnote{A random vector $\mathbf{x}$ is said to be isotropically distributed if, for any orthogonal matrix $\mathbf{A}$, $\mathbf{A}\mathbf{x}$ and $\mathbf{x}$ have the same distribution.}, $\forall N \in \mathbb{Z}_+$;\label{Assumption1-ii-b}
		\item $\lim_{N\rightarrow \infty}\|\mathbf{p}_k^N\|_2^2/N \overset{\text{a.s.}}{=}\tau_k^2$ for some $\tau_k > 0$, $\forall k\in[K]$;\label{Assumption1-ii-c}
	\end{enumerate}
\end{assumption}
Assumption \ref{linear-combination-assumption} models the mean-removed local updates $\{\tilde{\mathbf{g}}_m^N\}_{m=1}^M$ as linear combinations of a group of random vectors $\{\mathbf{p}_k^N\}_{k=1}^K$. The correlation between local updates comes from their shared base vectors. We note that one of the base vectors, $\mathbf{p}_1^N$, is allowed to be non-isotropic, and hence the mean-removed local updates $\{\tilde{\mathbf{g}}_m^N\}_{m=1}^M$ can possess a certain directional preference.

\begin{remark}
	We give the following justifications for Assumption \ref{linear-combination-assumption}:
	\begin{enumerate}
		\item  Recall that the local empirical loss function is written as a linear combination of a group of mutually independent components (i.e., the sample-wise loss functions), where the randomness of these components comes from the randomness of the sample generation process. This linear combination property tends to be inherited in model updates, for example, when the model updates are gradients of the local empirical loss functions. This justifies our linear combination assumption to some extent.
		\item In Assumption \ref{linear-combination-assumption}, all local updates are assumed to be weighted sums of the non-isotropic random vector $\mathbf{p}_1^N$ and some isotropic noise vectors $\{\mathbf{p}_k^N\}_{k=2}^K$. That is, if a local update possesses a certain directional preference, this directional preference come only from the base vector $\mathbf{p}_1^N$. This is consistent with a property of federated learning: every device hopes to update model parameters in the direction of the global update.
		\item Assumption \ref{linear-combination-assumption} allows the noises of local updates (i.e., $\{\sum_{k=2}^K e_{m, k} \mathbf{p}_k^N\}_{m=1}^M$) to be correlated by sharing common base vectors. In practice, the correlated noises are originated from common sources of randomness during the data collection/measurement.
		\item The well-known Gaussian Chief Executive Officer (CEO) model \cite{Berger1996The}, which is often used to model correlations in distributed systems\footnote{The CEO model has been used to model the correlation between local updates in \cite{Abdi2019Reducing}.}, is included as a special case of Assumption \ref{linear-combination-assumption}. In Gaussian CEO, local updates are modeled as $\tilde{\mathbf{g}}_m^{N} = \mathbf{p}_1^N + \mathbf{p}_{m+1}^N$, $\forall m \in [M]$, where $\{\mathbf{p}_1^N\} \cup \{\mathbf{p}_{m+1}^N\}_{m=1}^M$ are mutually independent random vectors with independent Gaussian elements. In our setting, $\mathbf{p}_1^N$ can be any random vector satisfying Assumptions \ref{linear-combination-assumption}-(\ref{Assumption1-ii-a}) and \ref{linear-combination-assumption}-(\ref{Assumption1-ii-c}), which covers a much broader class of distributions than independent Gaussian. 
		In particular, the elements of the local updates are generally correlated \cite{Xue2022FedOComp, Zhong2021Over}. Our model allows for such correlations, in contrast to the Gaussian CEO model. 
	\end{enumerate}
\end{remark}

\begin{proposition} \label{empirically-Gaussian}
	Consider a sequence of Haar distributed matrices $\{\mathbf{A}^{N}\in \mathbb{R}^{N \times N}\}_{N\in \mathbb{Z}_+}$ and $M$ sequences of mean-removed vectors $\{\tilde{\mathbf{g}}_m^N \in \mathbb{R}^N \}_{N \in \mathbb{Z}_+}$, $m \in [M]$. Assume that Assumption \ref{linear-combination-assumption} holds, and denote $\sigma_{m_1, m_2}^2 = \lim_{n\rightarrow \infty} \mathbb{E}[(\tilde{\mathbf{g}}_{m_1}^N)^\top \tilde{\mathbf{g}}_{m_2}^N] / N$, $\forall m_1,m_2 \in [M]$. Then, there exist $M$ sequences of random vectors $\{\tilde{\mathbf{x}}_m^N \in \mathbb{R}^N\}_{N\in \mathbb{Z}_+}$, $\forall m \in [M]$, such that
	\begin{enumerate}[(i)]
		\item $\tilde{\mathbf{x}}_m^N \sim \mathcal{N}\left(\mathbf{0}, \sigma_{m,m}^2 \mathbf{I}_N\right)$, $\forall m\in[M]$, $N\in \mathbb{Z}_+$;\label{Prop1-(i)}
		\item $[\tilde{\mathbf{x}}_1^N]_n, [\tilde{\mathbf{x}}_2^N]_n, \dots, [\tilde{\mathbf{x}}_M^N]_n$ are jointly Gaussian with $\mathbb{E}\left[[\tilde{\mathbf{x}}_{m_1}^N]_n\cdot[\tilde{\mathbf{x}}_{m_2}^N]_n\right] = \sigma_{m_1, m_2}^2$, $\forall m_1,m_2 \in [M]$, $n\in[N]$, $N\in \mathbb{Z}_+$;\label{Prop1-(ii)}
		\item $\lim_{N\rightarrow \infty}\left\| \mathbf{A}^{N} \tilde{\mathbf{g}}_m^{N} - \tilde{\mathbf{x}}_m^N\right\|^2_2 /N \overset{\text{a.s.}}{=} 0$, $\forall m \in [M]$.\label{Prop1-(iii)}
	\end{enumerate}
\end{proposition}

\begin{proof}[Proof]
	See Appendix \ref{proof-empirically-Gaussian}.
\end{proof}

By Proposition \ref{empirically-Gaussian}, under Assumption \ref{linear-combination-assumption}, for the quadratic distortion measure, the elements of each $\mathbf{x}_m^N$ in \eqref{random-rotation} can be asymptotically treated as i.i.d. zero-mean Gaussian variables, and $\{[\mathbf{x}_m^N]_n\}_{m=1}^M$ can be asymptotically treated as joint Gaussian variables with covariance matrix $\boldsymbol{\Sigma}_{X}$ satisfying $[\boldsymbol{\Sigma}_{X}]_{m_1,m_2}\!=\!\lim_{N\rightarrow \infty}\! \mathbb{E}[(\tilde{\mathbf{g}}_{m_1}^N)\!^\top \tilde{\mathbf{g}}_{m_2}^N] / N$, $\forall m_1,m_2 \in [M]$. This allows us to consider $\{[\mathbf{x}_1^N]_n, \dots, [\mathbf{x}_M^N]_n\}_{n=1}^N$ as $N$ samples generated from an $M$-component memoryless Gaussian source $(X_1, \dots, X_M) \sim \mathcal{N}(\boldsymbol{0},\boldsymbol{\Sigma}_{X})$,\footnote{Hereinafter, we also occasionally refer to the components of the $M$-component source as $M$ sources for convenience.}
and to apply the modified Berger-Tung coding to compress $\{\mathbf{x}_m^N\}_{m=1}^M$, as detailed in the next subsection.

\subsubsection{Data post-processing}
Given the output $\hat{\mathbf{x}}^N$ of the modified Berger-Tung decoder (which is an estimate of $\mathbf{x}^N \triangleq  \sum^M_{m=1}c_m \mathbf{x}^N_m$, as will be detailed in the next subsection), we perform the inverse process of data pre-processing to obtain an estimation of the global update $\mathbf{g}^N = \sum^M_{m=1}c_m \mathbf{g}^N_m$, i.e., the PS computes
\begin{equation}\label{g-hat}
\setlength{\abovedisplayskip}{3pt}
\setlength{\belowdisplayskip}{3pt}
\hat{\mathbf{g}}^{N} = (\mathbf{A}^{N})^{\top}\hat{\mathbf{x}}^{N} + \left(\sum_{m=1}^M c_m \bar{g}_m\right)\cdot\mathbf{1}_N.
\end{equation}

\subsection{Modified Berger-Tung Coding}
\label{Modified-Berger-Tung-Coding}
Under Assumption \ref{linear-combination-assumption}, sequence $\{[\mathbf{x}_1^N]_n, \dots, [\mathbf{x}_M^N]_n\}_{n=1}^N$ can be asymptotically treated as $N$ samples from an $M$-component memoryless Gaussian source $(X_1, \dots, X_M) \sim \mathcal{N}(\boldsymbol{0}, \boldsymbol{\Sigma}_{X})$. Since the local updates are correlated \cite{Abdi2019Reducing, Chen2020Scalecom, Zhong2021Over}, and only a function of the local updates needs to be recovered,
we modify Berger-Tung coding to compress $\{\mathbf{x}^N_m\}^M_{m=1}$.


Berger-Tung coding, as the achievability scheme of the Berger-Tung inner bound \cite{Berger1978Multiterminal, Tung1978MULTITERMINAL}, is a well-known random coding technique that non-cooperatively compresses multiple correlated information sources. Loosely speaking, in Berger-Tung coding, $M$ correlated sources $(X_1, \cdots, X_M)$ are encoded into codewords by $M$ separate encoders; based on the codewords, a joint decoder estimates all the sources $\{X_m\}_{m=1}^M$; the estimation performance is evaluated using $M$ preset distortion measures.


However, in our application, the decoder only aims to estimate a function $Y\triangleq \sum_{m=1}^M c_m X_m$ of the sources with quadratic distortion $\mathbb{E}[(Y-\hat{Y})^2]$, where $\hat{Y}$ denotes the output of the decoder.
To achieve this, we modify Berger-Tung coding by changing the constraints of codebook design to fit our application, i.e., changing the constraints from $\{\mathbb{E}[(X_m- \hat{X}_m)^2]\leq D_{BTC, m}\}_{m=1}^M$ to $\mathbb{E}[(Y- \hat{Y})^2]\leq D_{\text{MBTC}}$.\footnote{Please refer to \cite[Chapter 12]{El2011Network} for a detailed description.\label{Network-Information-Theory}} Following the achievability proof of the Berger-Tung inner bound\textsuperscript{\ref{Network-Information-Theory}}, any rate-distortion tuple in the following region can be achieved by our modified Berger-Tung coding (MBTC) scheme:

\begin{equation} \label{eq::MBTC_inner}
\begin{aligned} 
\mathcal{RD}_{\text{MBTC}} = 
\bigg\{ &
(R_{\textrm{MBTC}, 1}, \dots, R_{\textrm{MBTC},M}, D_\textrm{MBTC}):\text{$\exists$ $\hat{Y}:\mathbb{R}^M\rightarrow\mathbb{R}$ and $(U_1, \cdots, U_M)$, s.t.} \\
& \text{(i) } (\mathbf{x}_{/X_m}, \mathbf{u}_{/U_m}) \leftrightarrow X_m \leftrightarrow  U_m, \forall m \in [M]; \\ & \text{(ii) } \sum_{m\in\mathcal{S}}R_{\textrm{MBTC}, m} \geq I\left(\mathbf{x}^{\mathcal{S}}; \mathbf{u}^\mathcal{S} \mid \mathbf{u}^{\mathcal{S}^c}\right),~ \forall \mathcal{S} \subset [M], ~\mathcal{S} \neq \O;\\
& \text{(iii) } \sum_{m\in[M]}R_{\textrm{MBTC}, m} \geq I\left(\mathbf{x};\mathbf{u}\right);~~ \text{(iv) } D_\textrm{MBTC} \geq \mathbb{E}[(Y-\hat{Y}(\mathbf{u}))^2]
\bigg\},
\end{aligned}
\end{equation}
where $\mathbf{x} \triangleq [X_1, \dots, X_M]^{\top}$, $\mathbf{u} \triangleq [U_1, \dots, U_M]^{\top}$ denotes the auxiliary random vector, $\mathbf{x}_{/X_m}$ denotes the random vector obtained by deleting $X_m$ from $\mathbf{x}$, 
$I\left(\mathbf{x};\, \mathbf{u}\right)$ denotes the mutual information between $\mathbf{x}$ and $\mathbf{u}$, $I\left(\mathbf{x}^{\mathcal{S}};\, \mathbf{u}^\mathcal{S} \mid \mathbf{u}^{\mathcal{S}^c}\right)$ denotes the conditional mutual information between $\mathbf{x}^{\mathcal{S}}$ and $\mathbf{u}^\mathcal{S}$ given $\mathbf{u}^{\mathcal{S}^c}$, and the notation $(\mathbf{x}_{/X_m}, \mathbf{u}_{/U_m}) \leftrightarrow X_m \leftrightarrow U_m$ indicates that $(\mathbf{x}_{/X_m}, \mathbf{u}_{/U_m})$, $X_m$ and $U_m$ form a Markov chain in this order. 



\subsection{The Inner Region $\mathcal{RD}_{\mathrm{in}}$}
\label{summary}


Subsections \ref{Data-Processing} and \ref{Modified-Berger-Tung-Coding} have introduced the achievability scheme in Fig. \ref{MA-BT}. In this subsection, we derive the inner region $\mathcal{RD}_{\text{in}}$, a set of rate-distortion tuples that can be achieved using our proposed scheme.

Note that the expected aggregation distortion under our achievability scheme is given by
\begin{equation}
\label{the-same-distortion-measure1}
\begin{aligned}
\mathbb{E}[d(\mathbf{g}^N, \hat{\mathbf{g}}^N)]
= \frac{1}{N} \mathbb{E}\left[\left\|\sum_{m=1}^{M} c_m \mathbf{g}_m^N - \hat{\mathbf{g}}^N\right\|^2\right]
\stackrel{(a)}{=} \frac{1}{N} \mathbb{E}\left[\left\|\sum_{m=1}^{M} c_m \mathbf{x}_m^N - \hat{\mathbf{x}}^N\right\|^2\right],
\end{aligned}
\end{equation}
where the equality ($a$) follows from the orthogonality of matrix $\mathbf{A}^N$. The right-hand side of \eqref{the-same-distortion-measure1} happens to be the expected distortion of MBTC. Furthermore, from Fig. \ref{MA-BT}, the coding rates of MBTC are exactly the rates of our achievability scheme. Together, we make the following key observation: any rate-distortion tuple in $\mathcal{RD}_{\text{MBTC}}$ is numerically identical to a rate-distortion tuple in $\mathcal{RD}_{\text{in}}$ and vice versa.
This yields $\mathcal{RD}_{\text{in}}$, given by the following proposition.

\begin{proposition} \label{achievable-region-of-MBTC-based-MA-proposition}
	Under Assumption \ref{linear-combination-assumption}, considering distortion measure $d(\mathbf{a}, \mathbf{b}) = \|\mathbf{a} - \mathbf{b}\|^2/N$ with $\mathbf{a}, \mathbf{b} \in \mathbb{R}^N$ and aggregation target function $\kappa(a_1, \dots,$ $a_M) = \sum_{m=1}^{M} c_m a_m$, an inner region of $\mathcal{RD}^{\star}$, $\mathcal{RD}_{\mathrm{in}}$, is given by
	\begin{equation} \label{achievable-region-of-MBTC-based-MA}
	\begin{aligned} 
	\mathcal{RD}_{\mathrm{in}}= 
	\bigg\{&
	(R_1, \dots, R_M, D): \exists \hat{Y}: \mathbb{R}^M \rightarrow \mathbb{R} ~\mathrm{and}~ (U_1, \dots, U_M),~ \mathrm{s.t.}~\\	&\mathrm{(i)}\ (\mathbf{x}_{/X_m}, \mathbf{u}_{/U_m}) \leftrightarrow X_m \leftrightarrow U_m,\ \forall m \in [M], ~\mathbf{x}\sim \mathcal{N}(\boldsymbol{0}, \boldsymbol{\Sigma}_{X}); \\
	&\mathrm{(ii)}\ \sum_{m\in\mathcal{S}}R_m \geq I\left(\mathbf{x}^{\mathcal{S}};\, \mathbf{u}^\mathcal{S} \mid \mathbf{u}^{\mathcal{S}^c}\right),\ \forall \mathcal{S} \subset [M],\ \mathcal{S} \neq \O;\ \\ &\mathrm{(iii)}\ \sum_{m\in[M]}R_m \geq I\left(\mathbf{x};\, \mathbf{u}\right);\ \mathrm{(iv)}\ D \geq \mathbb{E}\left[(Y- \hat{Y}(\mathbf{u}))^2\right]
	\bigg\},
	\end{aligned}
	\end{equation}
	where $\hat{Y}(\cdot)$ is termed the reconstruction function, $\mathbf{x} \triangleq [X_1, \dots, X_M]^{\top}$, $\mathbf{u} \triangleq [U_1, \dots, U_M]^{\top}$, $Y\triangleq \sum^M_{m=1}c_m X_m$, and $\boldsymbol{\Sigma}_{X} \!\in\! \mathbb{R}^{M \times M}$ satisfies 
	$[\boldsymbol{\Sigma}_{X}]_{m_1,m_2} \!=\!\lim_{N\rightarrow \infty}\! \mathbb{E}[(\tilde{\mathbf{g}}_{m_1}^N)\!^\top \tilde{\mathbf{g}}_{m_2}^N] / N$, $\forall m_1,m_2 \in[M].$
\end{proposition}

	\section{Aggregation Distortion Minimization}
\label{Coding-Parameters-Design}
In this section, we develop an algorithm to minimize the aggregation distortion.
Specifically, we first conduct a convergence analysis and show that the optimization of the convergence rate can be transformed into the problem of aggregation distortion minimization. Subsequently, we put forth an algorithm to solve the distortion minimization problem.

\subsection{FL Convergence Analysis}

We now conduct a convergence analysis to establish the relationship between the FL convergence rate and the aggregation distortion, i.e., the mean square error of the estimated global update. To this end, we set the local update to the model gradient with respect to the local dataset, i.e., $\mathbf{g}_{m}^{(t)} = \nabla L_m(\boldsymbol{\theta}^{(t)})$, $\forall m \in [M]$, and make the following standard assumptions \cite{Amiri2020Machine, Fan2021Temporal}:
\begin{assumption} \label{ass::convergence_condition}
	The global loss function $L$ is strongly convex with parameter $\omega$, and has Lipschitz gradient with parameter $\Omega$, i.e., for any $\mathbf{a},\mathbf{b}\in\mathbb{R}^N$, 
	\begin{gather}
	L(\mathbf{b}) \geq L(\mathbf{a}) + \left(\mathbf{b} - \mathbf{a}\right)^\top \nabla L(\mathbf{a}) + \omega \left\|\mathbf{b} - \mathbf{a}\right\|^2/2 \\
	\|\nabla L(\mathbf{b}) - \nabla L(\mathbf{a})\| \leq \Omega \|\mathbf{b} - \mathbf{a}\|. \end{gather}
\end{assumption}
Proposition \ref{FL-convergence} gives a characterization of the FL convergence performance.
\begin{proposition} \label{FL-convergence}
	Under Assumption \ref{ass::convergence_condition}, consider $\mathbf{g}_{m}^{(t)} = \nabla L_m(\boldsymbol{\theta}^{(t)})$, $\forall m \in [M]$, and set the learning rate $\eta = 1/\Omega$. After $T$ communication rounds,
	\begin{equation}\label{FL-convergence-equation}
	\begin{aligned}
	L(\boldsymbol{\theta}^{(T+1)}) - L(\boldsymbol{\theta}^{\star}) \leq \left(1 - \frac{\omega}{\Omega}\right)^{T+1}\left[L(\boldsymbol{\theta}^{(0)}) - L(\boldsymbol{\theta}^{\star})\right] + \frac{N}{2\Omega} \sum_{t=0}^T \left(1 - \frac{\omega}{\Omega}\right)^{T-t}  D^{(t)},
	\end{aligned}
	\end{equation}
	where $D^{(t)}$ denotes the aggregation distortion at the $t$-th round.
	
\end{proposition} 

\begin{proof}[Proof]
	See Appendix \ref{proof-FL-convergence}.
\end{proof}

\subsection{Problem Formulation}
\label{problem-formulation}
The upper bound in Proposition \ref{FL-convergence} is a monotonically increasing function of each $D^{(t)}$, implying that we can potentially improve the FL convergence performance by separately minimizing each $D^{(t)}$. Recall that any rate-distortion tuple in $\mathcal{RD}_{\mathrm{in}}$ is achievable. Thus the problem reduces to finding a rate-distortion tuple in $\mathcal{RD}_{\mathrm{in}}$ that minimizes the distortion $D^{(t)}$. However, there are two pending issues:
\begin{itemize}
	\item [(i)] $\mathcal{RD}_{\mathrm{in}}$ is intractable due to the arbitrariness of the choice of $(U_1, \dots, U_M)$ and $\hat{Y}(\cdot)$;
	\item [(ii)] The rates in the rate-distortion tuples are chosen to satisfy the bit constraints of the FL uplinks.
\end{itemize}
In what follows, we first give a tractable inner region of $\mathcal{RD}_{\mathrm{in}}$ in Subsection \ref{A-Tractable-Inner-Region}, and then formulate an aggregation distortion minimization problem with link budget constraints in Subsection \ref{Optimization_problem_formulation}.

\subsubsection{A tractable inner region of $\mathcal{RD}_{\mathrm{in}}$}
\label{A-Tractable-Inner-Region}

We give a tractable inner region of $\mathcal{RD}_{\mathrm{in}}$ by picking $M$ parameterized auxiliary random variables and a specific reconstruction function for $\mathcal{RD}_{\mathrm{in}}$. For Gaussian equivalent sources $\{X_m\}^M_{m=1}$, setting the auxiliary variables to be Gaussian is a common choice \cite{El2011Network, Wagner2008Rate, Krithivasan2009Lattices}. Specifically, we define mutually independent random variables $\{V_m \sim \mathcal{N}(0, q_m)\}_{m=1}^M$ independent of $\{X_m\}_{m=1}^M$. Then we set the auxiliary random variable
\begin{equation}\label{U}
U_m = X_m + V_m, \, m \in [M].
\end{equation}
Since the quadratic distortion measure is considered, we set the reconstruction function as the minimum mean squared error (MMSE) estimator, i.e.,
\begin{equation}\label{Y_hat}
\hat{Y}(\mathbf{u}) = \mathbb{E}\left[Y \mid \mathbf{u}\right] = \mathbf{c}^\top \boldsymbol{\Sigma}_X (\boldsymbol{\Sigma}_X + \boldsymbol{\Sigma}_V)^{-1}\mathbf{u},
\end{equation}
where $\mathbf{c} \triangleq [c_1, \dots, c_M]^{\top}$ and $\boldsymbol{\Sigma}_{V} \triangleq \mathrm{diag}(\mathbf{q})$, $\mathbf{q} \triangleq [q_1, \dots, q_M]^{\top}$. Then, we derive closed form expressions of $I\left(\mathbf{x}^{\mathcal{S}};\, \mathbf{u}^\mathcal{S} \mid \mathbf{u}^{\mathcal{S}^c}\right)$, $I\left(\mathbf{x};\, \mathbf{u}\right)$ and $\mathbb{E}[(Y- \hat{Y}(\mathbf{u}))^2]$ to obtain an inner region of $\mathcal{RD}_{\textrm{in}}$:
\begin{align} 
\label{Guassian-Berger–Tung-inner-region-linear-function}
\mathcal{RD}_{\mathrm{in}}^{\mathrm{G}} = \bigcup\limits_{\mathbf{q} \in \mathbb{R}^{M}_+} \bigg\{
&(R_1, \dots, R_M, D):\ \sum_{m\in\mathcal{S}}R_m \geq I\left(\mathbf{x}^{\mathcal{S}};\, \mathbf{u}^\mathcal{S} \mid \mathbf{u}^{\mathcal{S}^c}\right),\ \forall \mathcal{S} \subset [M],\, \mathcal{S} \neq \O;\nonumber \\
&\sum_{m\in[M]}R_m \geq I\left(\mathbf{x};\, \mathbf{u}\right);
\ D \geq v(\mathbf{q})\bigg\},
\end{align}
where 
\begin{gather}
\setlength{\abovedisplayskip}{3pt}
\label{mutual-information-1}
I\left(\mathbf{x}^{\mathcal{S}};\, \mathbf{u}^\mathcal{S} \mid \mathbf{u}^{\mathcal{S}^c}\right) = \frac{1}{2}\log\left(\frac{\det\left(\boldsymbol{\Sigma}_{X} + \boldsymbol{\Sigma}_{V}\right)}{\det\left(\boldsymbol{\Sigma}_{X}^{\mathcal{S}^c} + \boldsymbol{\Sigma}_{V}^{\mathcal{S}^c}\right) \det\left(\boldsymbol{\Sigma}_{V}^{\mathcal{S}}\right)}\right),
\forall \mathcal{S} \subset [M],\, \mathcal{S} \neq \O, \\
\label{mutual-information-2}
I\left(\mathbf{x};\, \mathbf{u}\right) = \frac{1}{2}\log\left(\frac{\det\left(\boldsymbol{\Sigma}_{X} + \boldsymbol{\Sigma}_{V}\right)}{\det\left(\boldsymbol{\Sigma}_{V}\right)}\right),\\
\label{eq::form_of_v}
v(\mathbf{q})
\triangleq \mathbf{c}^\top \boldsymbol{\Sigma}_{X} \mathbf{c} - \mathbf{c}^\top\boldsymbol{\Sigma}_{X} \left(\boldsymbol{\Sigma}_{X} + \boldsymbol{\Sigma}_{V}\right)^{-1} \boldsymbol{\Sigma}_{X}^\top\mathbf{c}.
\end{gather}
We emphasize that $\mathcal{RD}_{\mathrm{in}}^{\mathrm{G}} \subset \mathcal{RD}_{\textrm{in}}$ for general $\boldsymbol{\Sigma}_{X}$ and $\mathbf{c}$. Since different choices of $\mathbf{q}$ lead to different codebooks of MBTC, we term $\mathbf{q}$ \emph{MBTC parameters}.

\subsubsection{Optimization problem formulation}
\label{Optimization_problem_formulation}
We now formulate an optimization problem to search the rate-distortion tuple in $\mathcal{RD}_{\mathrm{in}}^{\mathrm{G}}$ with minimum distortion $D$.
From \eqref{Guassian-Berger–Tung-inner-region-linear-function}, for given MBTC parameters, the minimum distortion is given by $v(\mathbf{q})$. Thus the problem reduces to finding a minimal $v(\mathbf{q})$ in $\mathcal{RD}_{\mathrm{in}}^{\mathrm{G}}$ by tuning the MBTC parameters $\mathbf{q}$.

As mentioned before, $\mathbf{q}$ cannot be arbitrarily chosen due to the link bit constraints. 
Let $r_m^{tot}$ denote the maximum number of total bits that device $m$ can transmit to the PS. To ensure reliable uplink transmission, the (source coding) rates in (\ref{Guassian-Berger–Tung-inner-region-linear-function}) need to satisfy
\begin{equation}\label{channel-capacity-limit-new}
R_m \leq r_m, \ \forall m\in [M],
\end{equation}
where $r_m \triangleq r_m^{tot} / N$. Combining \eqref{Guassian-Berger–Tung-inner-region-linear-function}, \eqref{eq::form_of_v}, \eqref{channel-capacity-limit-new} and the above discussion, the distortion minimization problem is formulated as
\begin{subequations}
	\setlength{\abovedisplayskip}{3pt}
	\label{original-optimization-problem}
	\begin{align}
	\label{eq::MBTC_opt_objective}
	\max_{\mathbf{q}} \quad &\mathbf{c}^\top\boldsymbol{\Sigma}_{X} \left(\boldsymbol{\Sigma}_{X} + \boldsymbol{\Sigma}_{V}\right)^{-1} \boldsymbol{\Sigma}_{X}^\top\mathbf{c} \\
	\label{eq::MBTC_opt_cond1}
	\mathrm{s.t.} \quad
	&I\left(\mathbf{x}^{\mathcal{S}};\, \mathbf{u}^\mathcal{S} \mid \mathbf{u}^{\mathcal{S}^c}\right) \leq \sum_{m\in\mathcal{S}}r_m, \ \forall \mathcal{S} \subset [M],\ \mathcal{S} \neq \O, \\
	\label{eq::MBTC_opt_cond2}
	&I\left(\mathbf{x};\, \mathbf{u}\right) \leq \sum_{m\in[M]}r_m. 
	\end{align}
\end{subequations}



	\subsection{MBTC Optimization Algorithm}
\label{Optimization-Algorithm-for-the-General-Case}

In this subsection, we propose an iterative algorithm based on majorization-minimization (MM) to solve problem \eqref{original-optimization-problem}. Our algorithm starts with a feasible point called the \emph{current-point}. Each iteration round consists of two steps. In the first step, we construct a \emph{surrogate problem}, whose objective serves as a lower bound of the original objective with equality holds at the current-point. 
Besides, the feasible region of the surrogate problem should be a subset of the original feasible region and contains the current-point.
In the second step, we solve the surrogate problem, and the solution will be used as the current-point in the next iteration.

Before proceeding, we present two lemmas for constructing the surrogate problem. Specifically, Lemma \ref{objective-function-lower-bound} helps find a lower bound of the original objective \eqref{eq::MBTC_opt_objective}, and Lemma \ref{mutual-information-upper-bound} helps find a subset of the original feasible region.

\begin{lemma}[{\cite[Theorem 2]{Shen2018Fractional}}] \label{objective-function-lower-bound}
	For any $\mathbf{a}, \mathbf{b} \in \mathbb{R}^{M}$ and positive definite matrix $\mathbf{B} \in \mathbb{R}^{M \times M}$,
	\begin{equation}
	\mathbf{a}^{\top}\mathbf{B}^{-1}\mathbf{a} \geq 2\mathbf{a}^{\top}\mathbf{b} - \mathbf{b}^{\top}\mathbf{B}\mathbf{b},
	\end{equation}
	where the equality holds when $\mathbf{a}$, $\mathbf{b}$ and $\mathbf{B}$ satisfy $\mathbf{b} = \mathbf{B}^{-1}\mathbf{a}$.
\end{lemma} 

\begin{proof}[Proof]
	$\mathbf{a}^{\top}\mathbf{B}^{-1}\mathbf{a}
	\geq
	\mathbf{a}^{\top}\mathbf{B}^{-1}\mathbf{a}
	-
	(\mathbf{b} - \mathbf{B}^{-1}\mathbf{a})^{\top}\mathbf{B}(\mathbf{b} - \mathbf{B}^{-1}\mathbf{a})
	=
	2\mathbf{a}^{\top}\mathbf{b} - \mathbf{b}^{\top}\mathbf{B}\mathbf{b}$,
	where the equality holds when $\mathbf{b} = \mathbf{B}^{-1}\mathbf{a}$.
\end{proof}

\begin{lemma} \label{mutual-information-upper-bound}
	
	Let $\mathbf{x} \sim \mathcal{N}(\mathbf{0}, \boldsymbol{\Sigma}_{X})$ and $\mathbf{v} \sim \mathcal{N}(\mathbf{0}, \boldsymbol{\Sigma}_{V})$ with diagonal $\boldsymbol{\Sigma}_{V}$. Denote $\mathbf{u} = \mathbf{x} + \mathbf{v} \in \mathbb{R}^M$. Given a nonempty $\mathcal{S} \subset [M]$, for any $\mathbf{E} \in \mathbb{R}^{|\mathcal{S}| \times |\mathcal{S}^c|}$ and $\mathbf{F} \in \mathbb{R}^{|\mathcal{S}| \times |\mathcal{S}|}$ with $\mathbf{F} \succ 0$, we have
	\begin{equation}\label{mutual-information-upper-bound-one}
	I\left(\mathbf{x}^{\mathcal{S}};\, \mathbf{u}^\mathcal{S} \mid \mathbf{u}^{\mathcal{S}^c}\right)
	\leq \chi_{\mathcal{S}}(\mathbf{E}, \mathbf{F}, \boldsymbol{\Sigma}_{V})
	+ \xi_{\mathcal{S}}(\mathbf{E}, \mathbf{F}),
	\end{equation}
	where
	\begin{gather}
	\setlength{\abovedisplayskip}{3pt}
	\label{chi-S-one}
	\chi_{\mathcal{S}}(\mathbf{E}, \mathbf{F}, \boldsymbol{\Sigma}_{V}) = \frac{\log(e)}{2}  
	\mathrm{tr}\left\{\mathbf{F}^{-1} \boldsymbol{\Sigma}_{V}^{\mathcal{S}}\right\}
	+\frac{\log(e)}{2}  
	\mathrm{tr}\left\{\mathbf{E}^{\top}\mathbf{F}^{-1}\mathbf{E} \boldsymbol{\Sigma}_{V}^{\mathcal{S}^c}\right\} 
	-
	\frac{1}{2}\log\left(\mathrm{det}(\boldsymbol{\Sigma}_{V}^{\mathcal{S}})\right),\\
	\label{xi-S-one}
	\xi_{\mathcal{S}}(\mathbf{E}, \mathbf{F}) \!=\! \frac{1}{2}\log\left(\mathrm{det}(\mathbf{F})\right)
	\!+\!
	\frac{\log(e)}{2} \mathrm{tr}\left\{\mathbf{F}^{-1}\left( \boldsymbol{\Sigma}_{X}^{\mathcal{S}}
	\!+\!\mathbf{E}\boldsymbol{\Sigma}_{X}^{\mathcal{S}^c}\mathbf{E}^{\top}
	\!-\!\mathbf{E}\boldsymbol{\Sigma}_{X}^{\mathcal{S}^c, \mathcal{S}}
	\!-\!\boldsymbol{\Sigma}_{X}^{\mathcal{S}, \mathcal{S}^c}\mathbf{E}^{\top}
	\right)\right\} \!-\! \frac{|\mathcal{S}|\log(e)}{2},
	\end{gather}
	and the equality holds when $\mathbf{E} = \boldsymbol{\Sigma}_{X}^{\mathcal{S}, \mathcal{S}^c}(\boldsymbol{\Sigma}_{X}^{\mathcal{S}^c} + \boldsymbol{\Sigma}_{V}^{\mathcal{S}^c})^{-1}$ and $\mathbf{F} = \boldsymbol{\Sigma}_{X}^{\mathcal{S}} + \boldsymbol{\Sigma}_{V}^{\mathcal{S}} - \boldsymbol{\Sigma}_{X}^{\mathcal{S}, \mathcal{S}^c}
	(\boldsymbol{\Sigma}_{X}^{\mathcal{S}^c} + \boldsymbol{\Sigma}_{V}^{\mathcal{S}^c})^{-1}
	\boldsymbol{\Sigma}_{X}^{\mathcal{S}^c, \mathcal{S}}$. Similarly, for any $\mathbf{G} \in \mathbb{R}^{M \times M}$ with $\mathbf{G} \succ 0$, we have
	\begin{equation}\label{mutual-information-upper-bound-two}
	\begin{aligned} 
	I\left(\mathbf{x};\, \mathbf{u}\right)
	\leq
	\chi_{[M]}(\mathbf{G}, \boldsymbol{\Sigma}_{V})
	+ \xi_{[M]}(\mathbf{G}),
	\end{aligned}
	\end{equation}
	where
	\begin{equation}\label{chi-M-one}
	\chi_{[M]}(\mathbf{G}, \boldsymbol{\Sigma}_{V}) = \frac{\log(e)}{2}  
	\mathrm{tr}\left\{\mathbf{G}^{-1} \boldsymbol{\Sigma}_{V}\right\}
	-
	\frac{1}{2}\log\left(\mathrm{det}(\boldsymbol{\Sigma}_{V})\right),
	\end{equation}
	\begin{equation}\label{xi-M-one}
	\xi_{[M]}(\mathbf{G}) = \frac{1}{2}\log\left(\mathrm{det}(\mathbf{G})\right)
	+
	\frac{\log(e)}{2} \mathrm{tr}\left\{\mathbf{G}^{-1} \boldsymbol{\Sigma}_{X}\right\} - \frac{M\log(e)}{2},
	\end{equation}
	and the equality holds when $\mathbf{G} = \boldsymbol{\Sigma}_{X} + \boldsymbol{\Sigma}_{V}$.

\end{lemma} 

\begin{proof}[Proof]
	See Appendix \ref{proof-mutual-information-upper-bound}.
\end{proof}

Recall that $\boldsymbol{\Sigma}_{V} = \mathrm{diag}(\mathbf{q})$. Given a feasible point $\boldsymbol{\Sigma}_{V} = \widehat{\boldsymbol{\Sigma}}_{V}$, we construct a problem as
\begin{subequations} \label{approximation-optimization-problem-two}
	\begin{align}
	\max_{\boldsymbol{\Sigma}_V} \quad 
	&2\mathbf{c}^\top\boldsymbol{\Sigma}_{X} \mathbf{b} - \mathbf{b}^{\top}(\boldsymbol{\Sigma}_X + \boldsymbol{\Sigma}_{V})\mathbf{b} \label{surrogate-objective}\\
	\mathrm{s.t.} \quad
	&\chi_{\mathcal{S}}(\mathbf{E}_{\mathcal{S}}, \mathbf{F}_{\mathcal{S}}, \boldsymbol{\Sigma}_{V})
	+ \xi_{\mathcal{S}}(\mathbf{E}_{\mathcal{S}}, \mathbf{F}_{\mathcal{S}})
	\leq \sum_{m\in\mathcal{S}}r_m,\,\forall \mathcal{S} \subset [M],\, \mathcal{S} \neq \O, \label{surrogate-constraint1}\\
	&\chi_{[M]}(\mathbf{G}, \boldsymbol{\Sigma}_{V})
	+ \xi_{[M]}(\mathbf{G})
	\leq \sum_{m\in[M]}r_m,\label{surrogate-constraint2}
	\end{align}
\end{subequations}
where $\mathbf{b} = (\boldsymbol{\Sigma}_{X} + \widehat{\boldsymbol{\Sigma}}_{V})^{-1}\boldsymbol{\Sigma}_{X}\mathbf{c}$, $\mathbf{E}_{\mathcal{S}} = \boldsymbol{\Sigma}_{X}^{\mathcal{S}, \mathcal{S}^c}(\boldsymbol{\Sigma}_{X}^{\mathcal{S}^c} + \widehat{\boldsymbol{\Sigma}}_{V}^{\mathcal{S}^c})^{-1}$ and $\mathbf{F}_{\mathcal{S}} = \boldsymbol{\Sigma}_{X}^{\mathcal{S}} + \widehat{\boldsymbol{\Sigma}}_{V}^{\mathcal{S}} - \boldsymbol{\Sigma}_{X}^{\mathcal{S}, \mathcal{S}^c}
(\boldsymbol{\Sigma}_{X}^{\mathcal{S}^c} + \widehat{\boldsymbol{\Sigma}}_{V}^{\mathcal{S}^c})^{-1}
\boldsymbol{\Sigma}_{X}^{\mathcal{S}^c, \mathcal{S}}$ for all nonempty set $\mathcal{S}\subset [M]$, and $\mathbf{G} = \boldsymbol{\Sigma}_{X} + \widehat{\boldsymbol{\Sigma}}_{V}$. According to Lemma \ref{objective-function-lower-bound}, (\ref{surrogate-objective}) is a lower bound of \eqref{eq::MBTC_opt_objective} with the equality holds when $\boldsymbol{\Sigma}_{V} = \widehat{\boldsymbol{\Sigma}}_{V}$. From Lemma \ref{mutual-information-upper-bound}, the feasible region of \eqref{original-optimization-problem} contains the feasible region of (\ref{approximation-optimization-problem-two}) and both of them contain the point $\boldsymbol{\Sigma}_{V} = \widehat{\boldsymbol{\Sigma}}_{V}$. Thus \eqref{approximation-optimization-problem-two} is a surrogate problem of (\ref{original-optimization-problem}) and point $\boldsymbol{\Sigma}_{V} = \widehat{\boldsymbol{\Sigma}}_{V}$ is the current-point. 

It is not difficult to verify that (\ref{surrogate-objective}) is a linear function of $\mathbf{q}$, and both $\chi_{\mathcal{S}}(\mathbf{E}_{\mathcal{S}}, \mathbf{F}_{\mathcal{S}}, \boldsymbol{\Sigma}_{V})$ and $\chi_{[M]}(\mathbf{G}, \boldsymbol{\Sigma}_{V})$ are convex functions of $\mathbf{q}$. Thus the surrogate problem (\ref{approximation-optimization-problem-two}) is convex and can be solved optimally with existing convex optimization solvers such as CVXPY \cite{Diamond2016CVXPY}. By repeatedly constructing and solving this surrogate problem following the MM framework introduced before, we can finally obtain a suboptimal solution to the original problem (\ref{original-optimization-problem}). We summarize the proposed MM-based algorithm as Algorithm \ref{proposed-algorithm}. This algorithm converges since the objective value of the original problem (\ref{original-optimization-problem}) monotonically non-decreasing in the iterative process. 

Note that problem (\ref{approximation-optimization-problem-two}) has $2^M - 1$ constraints (inherited from problem (\ref{original-optimization-problem})), which increases exponentially with the device number $M$. Thus, when considering FL systems with a relatively large number of devices, Algorithm \ref{proposed-algorithm} becomes computationally prohibitive. In the next section, under some symmetry assumptions, we show that problem (\ref{original-optimization-problem}) can be reformulated into a form with much fewer constraints, allowing the development of more efficient algorithms.


\begin{algorithm}[h] 
	\caption{MBTC Optimization Algorithm}
	\label{proposed-algorithm}
	\begin{algorithmic}[1]
		\Require
		$\boldsymbol{\Sigma}_{X}$, $\mathbf{c}$, $\{r_m\}_{m=1}^M$.
		\Ensure
		solution $\mathbf{q}^{*}$.
		\State Initialize $\mathbf{q}^{(0)}$ to a feasible point of problem (\ref{original-optimization-problem}), let $\boldsymbol{\Sigma}_{V}^{(0)} = \mathrm{diag}(\mathbf{q}^{(0)})$, initialize iteration number $i=0$ and threshold $\epsilon > 0$.
		\Repeat
		\State $\widehat{\boldsymbol{\Sigma}}_{V} = \boldsymbol{\Sigma}_{V}^{(i)}$.
		\State Solve convex problem (\ref{approximation-optimization-problem-two}) to obtain $\boldsymbol{\Sigma}_{V}^{(i+1)}$.
		\State Update $i=i+1$.
		\Until{the fractional increase of the objective value of problem (\ref{original-optimization-problem}) is below the threshold $\epsilon$ or the maximum number of iterations is reached.}
		\State Set $[\mathbf{q}^*]_m = [\boldsymbol{\Sigma}_{V}^{(i)}]_{m,m}$, $\forall m \in [M]$.
	\end{algorithmic}
\end{algorithm}

	\section{Aggregation Distortion Minimization Under Symmetric Assumptions}
	\label{MBTC-Parameters-Optimization-with-Symmetric-Assumptions}
	In this section, we recast problem (\ref{original-optimization-problem}) into a form with far fewer constraints under certain symmetry assumptions, and then develop an optimization algorithm to solve it.
	
	\subsection{Problem Formulation under Symmetry Assumptions}
	The discussion in this section is based on the following three symmetry assumptions.
	\begin{assumption} \label{symmetry-assumption}
		The FL system is symmetric in the following senses:
		\begin{enumerate}[(i)]
			\item Symmetry of sources: The covariance matrix $\boldsymbol{\Sigma}_{X} = \rho \sigma_{X}^2 \mathbf{1} \mathbf{1}^{\top} + (\sigma_{X}^2 - \rho \sigma_{X}^2)\mathbf{I}$; \label{symmetry-assumption1}
			\item Symmetry of target coefficients: $\mathbf{c} = \lambda \cdot \mathbf{1}$ with $\lambda \neq 0$; \label{symmetry-assumption2}
			\item Symmetry of bit-constraints: All devices are divided into $J$ groups $\{\mathcal{G}_j\subset [M]\}^J_{j=1}$, and the devices in each group have the same bit-constraint, i.e., $r_m = r_{(j)}$, $\forall m \in \mathcal{G}_j$, $j \in [J]$. \label{symmetry-assumption3}
		\end{enumerate}
	\end{assumption}
	
	We first justify these assumptions. When data is i.i.d. among the devices and all devices have similar sample sizes, Assumption \ref{symmetry-assumption}-(\ref{symmetry-assumption1}) approximately holds, as shown in Fig. $4$ of \cite{Zhong2021Over}. Assumption \ref{symmetry-assumption}-(\ref{symmetry-assumption2}) can be satisfied by simply adjusting the aggregation target function.\footnote{For instance, when considering the function $\kappa(\mathbf{g}_1^N,\dots,\mathbf{g}_M^N) = \sum_{m=1}^M K_m\mathbf{g}_m^N/K$, Assumption \ref{symmetry-assumption}-(\ref{symmetry-assumption2}) can be satisfied by adjusting the sample size $K_m$ of every device to be the same.}
	Assumption \ref{symmetry-assumption}-(\ref{symmetry-assumption3}) is also not difficult to satisfy since devices with loose bit-constraints can accommodate devices with tighter bit-constraints by reducing the number of transmitted bits. 
	
	We next recast problem (\ref{original-optimization-problem}) under Assumption \ref{symmetry-assumption}. It can be verified that under Assumption \ref{symmetry-assumption}, problem \eqref{original-optimization-problem} is symmetric with respect to the optimization variables in the same group (defined in Assumption \ref{symmetry-assumption}-\eqref{symmetry-assumption3}).
	Thus the optimal solution of problem (\ref{original-optimization-problem}), $\{q_m^{\mathrm{opt}}\}_{m=1}^M$, satisfies $q^\mathrm{opt}_m = q_{(j)}^\mathrm{opt}$, $\forall m \in \mathcal{G}_j, \, j \in [J]$. Thus we can solve (\ref{original-optimization-problem}) by solving
	\begin{subequations} \label{original-optimization-problem-three}
		\begin{align}
		\max_{\{q_{(j)}\}_{j=1}^J} \quad 
		&\mathbf{c}^\top\boldsymbol{\Sigma}_X \left(\boldsymbol{\Sigma}_X + \boldsymbol{\Sigma}_V\right)^{-1} \boldsymbol{\Sigma}_X^\top\mathbf{c}\\
		\mathrm{s.t.} \quad
		&\text{\eqref{eq::MBTC_opt_cond1} and \eqref{eq::MBTC_opt_cond2} hold}, \label{original-optimization-problem-three-c_b}\\
		&q_m = q_{(j)}, \, \forall m \in \mathcal{G}_j, \, j \in [J]. \label{original-optimization-problem-three-c_d}
		\end{align}
	\end{subequations}
	Denote $M_j \triangleq |\mathcal{G}_j|$, $j \in [J]$. From Appendix \ref{Problem-Form-Transformation}, we can recast problem (\ref{original-optimization-problem-three}) as
	\begin{subequations} \label{original-optimization-problem-four}
		\begin{align}
		\max_{\{q_{(j)}\}_{j=1}^J} \quad 
		&\sum_{j=1}^J \frac{M_j}{q_{(j)} + (1-\rho) \sigma_{X}^2} \label{original-optimization-problem-four-obj}\\
		\mathrm{s.t.} \quad
		&\vartheta\left(\{q_{(j)}\}_{j=1}^J; \{\varsigma_j\}_{j=1}^J\right) \leq \sum_{j=1}^J \varsigma_j r_{(j)},\, \forall \varsigma_j \in \{0\}\cup [M_j], j \in [J], \prod_{j=1}^J \varsigma_j \neq 0,\label{original-optimization-problem-four-constraints}
		\end{align}
	\end{subequations}
	where 
	\begin{equation}\label{a_q}
	\begin{aligned} 
	\vartheta\left(\{q_{(j)}\}_{j=1}^J; \{\varsigma_j\}_{j=1}^J\right)
	\triangleq &\frac{1}{2} \sum_{j=1}^J \varsigma_j\log\left(1 \!+\! \frac{(1\!-\!\rho)\sigma_{X}^2}{q_{(j)}}\right)
	\!+\! \frac{1}{2} \log\left(1 \!+\! \sum_{j=1}^J \frac{M_j \rho \sigma_{X}^2}{(1\!-\!\rho)\sigma_{X}^2 + q_{(j)}}\right)\\
	&- \frac{1}{2} \log\left(1 + \sum_{j=1}^J \frac{(M_j - \varsigma_j) \rho \sigma_{X}^2}{(1-\rho)\sigma_{X}^2 + q_{(j)}}\right).
	\end{aligned}
	\end{equation}
	The solution $\{q_{(j)}^*\}_{j=1}^J$ of \eqref{original-optimization-problem-four} gives rise to a solution of \eqref{original-optimization-problem} i.e., $q^{*}_m = q_{(j)}^*$, $\forall m \in \mathcal{G}_j$, $j \in [J]$.
	
	\subsection{MBTC Optimization Algorithm Under Symmetry Assumptions}
	We now develop an MM-based algorithm to solve problem (\ref{original-optimization-problem-four}). Specifically, to construct a surrogate problem, we need to find a lower bound of the objective (\ref{original-optimization-problem-four-obj}) and an upper bound of function $\vartheta$. 
	Note that functions $\log(1 + a / x)$, $a \geq 0$, $x \textgreater 0$ and $\log(1 + \sum_{j=1}^J a_j / (b + x_j))$, $b \geq 0$, $a_j \geq 0$, $x_j \textgreater 0$, $\forall j \in [J]$, are both convex. Thus an appropriate upper bound of function $\vartheta$ can be obtained by expanding its third term as its first-order Taylor polynomial. Specifically, given a feasible point $\{\hat{q}_{(j)}\}_{j=1}^J$, we have
	\begin{subequations}
		\begin{equation}\label{lower-bound}
		\begin{aligned} 
		\sum_{j=1}^J \frac{M_j}{q_{(j)} + (1-\rho) \sigma_{X}^2} \geq 
		\sum_{j=1}^J \frac{M_j}{\hat{q}_{(j)} + (1-\rho) \sigma_{X}^2}
		- \sum_{j=1}^J \frac{M_j}{\left(\hat{q}_{(j)} + (1-\rho) \sigma_{X}^2\right)^2}\left(q_{(j)} - \hat{q}_{(j)}\right),
		\end{aligned}
		\end{equation}
		and $\vartheta\left(\{q_{(j)}\}_{j=1}^J; \{\varsigma_j\}_{j=1}^J\right)$ is upper-bounded by
		\begin{multline}\label{upper-bound}
		\vartheta^{\text{up}}\left(\{q_{(j)}\}_{j=1}^J; \{\varsigma_j\}_{j=1}^J\right) \triangleq 
		\frac{1}{2} \sum_{j=1}^J \varsigma_j\log\left(1 + \frac{(1-\rho)\sigma_{X}^2}{q_{(j)}}\right)\\
		+ \frac{1}{2} \log\left(1 + \sum_{j=1}^J \frac{M_j \rho \sigma_{X}^2}{(1-\rho)\sigma_{X}^2 + q_{(j)}}\right)
		-\frac{1}{2} \log\left(1 + \sum_{j=1}^J \frac{(M_j - \varsigma_j) \rho \sigma_{X}^2}{(1-\rho)\sigma_{X}^2 + \hat{q}_{(j)}}\right)\\
		+\frac{1}{2} \frac{\log(e)}{\left(1 + \sum_{i=1}^J \frac{(M_i - \varsigma_i) \rho \sigma_{X}^2}{(1-\rho)\sigma_{X}^2 + \hat{q}_{(i)}}\right)}
		\sum_{j=1}^J \frac{(M_j - \varsigma_j) \rho \sigma_{X}^2}{\left((1-\rho)\sigma_{X}^2 + \hat{q}_{(j)}\right)^2} \left(q_{(j)} - \hat{q}_{(j)}\right),
		\end{multline}
	\end{subequations}
	where both the equalities in \eqref{lower-bound} and \eqref{upper-bound} hold when $q_{(j)} = \hat{q}_{(j)}$, $\forall j \in [J]$. Further, note that the maximization of the right-hand side with respect to $\{q_{(j)}\}^J_{j=1}$ of \eqref{lower-bound} is equivalent to the minimization of the term $\sum_{j=1}^J M_jq_{(j)}/(\hat{q}_{(j)} + (1-\rho) \sigma_{X}^2)^2$. Thus, a surrogate problem is given by
	\begin{subequations} \label{approximation-optimization-problem-three}
		\begin{align}
		\min_{\{q_{(j)}\}_{j=1}^J} \quad 
		&\sum_{j=1}^J \frac{M_j q_{(j)}}{\left(\hat{q}_{(j)} + (1-\rho) \sigma_{X}^2\right)^2}\\
		\mathrm{s.t.} \quad
		&\vartheta^{\text{up}}\left(\{q_{(j)}\}_{j=1}^J; \{\varsigma_j\}_{j=1}^J\right) \leq \sum_{j=1}^J \varsigma_j r_{(j)}, ~ \forall \varsigma_j \in \{0\}\cup [M_j],~\forall j \in [J], ~\prod_{j=1}^J \varsigma_j \neq 0.
		\end{align}
	\end{subequations}
	
	Since $\vartheta^{\text{up}}\left(\{q_{(j)}\}_{j=1}^J; \{\varsigma_j\}_{j=1}^J\right)$ is convex, problem (\ref{approximation-optimization-problem-three}) is convex and can be solved optimally by existing convex optimization solvers such as CVXOPT \cite{Vandenberghe2010The}. Again, by repeatedly constructing and solving this surrogate problem, we can finally obtain a suboptimal solution to problem (\ref{original-optimization-problem-four}). We summarize this algorithm as Algorithm \ref{proposed-algorithm-for-grouping-case}. This algorithm converges since the objective value of problem (\ref{original-optimization-problem-four}) is monotonically non-decreasing in the iterative process.

	Problem (\ref{approximation-optimization-problem-three}) has $\prod_{j=1}^J (M_j + 1) - 1$ constraints, with growth rate much slower than that of (\ref{approximation-optimization-problem-two}), which is $2^M - 1$. Thus Algorithm \ref{proposed-algorithm-for-grouping-case} can be used to optimize the MBTC parameters for larger-scale FL systems. Clearly, the choice of group number $J$ gives rise to a trade-off between the computational complexity for solving \eqref{approximation-optimization-problem-three} and the system performance. Specifically, if $J=1$, problem (\ref{approximation-optimization-problem-three}) has only $M$ constraints. This greatly reduces the computational complexity compare with Algorithm \ref{proposed-algorithm}, but may severely sacrifice the aggregation accuracy in order to satisfy Assumption \ref{symmetry-assumption}-(\ref{symmetry-assumption3}). As $J$ approaches $M$, the computational complexity gradually catches up with that of problem \eqref{original-optimization-problem}, while the degradation of aggregation accuracy also diminishes. In practice, we can flexibly determine the value of $J$ as needed.

	\begin{algorithm}[h] 
		\caption{MBTC Optimization Algorithm under Symmetry Assumptions}
		\label{proposed-algorithm-for-grouping-case}
		\begin{algorithmic}[1]
			\Require
			$\rho$, $\sigma_{X}^2$, $\{\mathcal{G}_j\}_{j=1}^J$, $\{r_{(j)}\}_{j=1}^J$.
			\Ensure
			solution $\mathbf{q}^*$.
			\State Set $M_j = |\mathcal{G}_j|$, $\forall j \in [J]$, initialize $\{q_{(j)}^{(0)}\}_{j=1}^J$ to a feasible point of problem (\ref{original-optimization-problem-four}), initialize iteration number $i=0$ and threshold $\epsilon > 0$.
			\Repeat
			\State $\hat{q}_{(j)} = q_{(j)}^{(i)}$, $j\in[J]$.
			\State Solve convex problem (\ref{approximation-optimization-problem-three}) to obtain $\{q_{(j)}^{(i+1)}\}_{j=1}^J$.
			\State Update $i=i+1$.
			\Until{the fractional increase of the objective value of problem (\ref{original-optimization-problem-four}) is below the threshold $\epsilon$ or the maximum number of iterations is reached.}
			\State Set $[\mathbf{q}^*]_m = q_{(j)}^*$, $\forall m \in \mathcal{G}_j, \forall j \in [J]$.
		\end{algorithmic}
	\end{algorithm}

	\section{Numerical Results} \label{Numerical-Results}
	
	In this section, we first introduce the method to numerically evaluate the limits of FL convergence performance, then reveal the gap between the baseline schemes and our theoretical bound in terms of aggregation distortion, convergence performance, and communication cost.

	\subsection{FL Convergence Performance Evaluation}
	\label{FL_Convergence_Performance_Evaluation}
	
	Note that the solution $\mathbf{q}^*$ obtained by solving problem \eqref{original-optimization-problem} not only corresponds to a point in $\mathcal{RD}_{\mathrm{in}}^{\mathrm{G}}$ with small distortion, but also implies a codebook generation method for our achievability scheme. At each iteration round $t$, we can apply our achievability scheme with the $\mathbf{q}^*$-codebook to compute $\hat{\mathbf{g}}^{(t)}$ and then use it to update the global model. In this way, the FL convergence performance limits (in the sense of the performance limits of model aggregation) can be numerically evaluated. In the following, we detail the convergence performance evaluation method.
	
	%
	Recall that the FL training follows the four steps in Subsection \ref{Federated Learning System}. Since the steps (i), (ii), and (iv) are straightforward, we focus on the model aggregation step, i.e., how to compute $\hat{\mathbf{g}}^N$ using $\{\mathbf{g}_m^N\}_{m=1}^M$. Specifically, given the local updates $\{\mathbf{g}_m^N\}_{m=1}^M$ at iteration round $t$, we first compute $\{\mathbf{x}_m^N\}_{m=1}^M$ using \eqref{mean_removal} and \eqref{random-rotation}\footnote{In the simulation, as an approximation to the random rotation step \eqref{random-rotation}, we divide each vector $\tilde{\mathbf{g}}_m^N$ into segments with a length of $1024$ and then generate Haar matrices to multiply the segments (the remaining segment with less than $1024$ elements is multiplied by a Haar matrix with the corresponding dimension).}.
	Then, we follow Section \ref{A-Tractable-Inner-Region} for the encoding and decoding of MBTC.
	Specifically, to calculate the output of the decoder, we first approximate $\boldsymbol{\Sigma}_{X}$ by $[\boldsymbol{\Sigma}_{X}]_{m_1,m_2} \approx (\tilde{\mathbf{g}}_{m_1}^N)^\top \tilde{\mathbf{g}}_{m_2}^N / N$, $\forall m_1,m_2 \in [M]$, then solve problem (\ref{original-optimization-problem}) to obtain the MBTC parameters $\mathbf{q}^*=[q_1^*,\dots,q_M^*]^\top$.
	Denoting $\boldsymbol{\Sigma}_{V}^* \triangleq \mathrm{diag}(\mathbf{q}^*)$, we combine \eqref{U} and \eqref{Y_hat} to calculate the output of the decoder as
	\begin{equation}
	\label{calculate-decoder-output}
	\begin{aligned}
	[\hat{\mathbf{x}}^N]_n = \mathbf{c}^\top \boldsymbol{\Sigma}_X (\boldsymbol{\Sigma}_X + \boldsymbol{\Sigma}_V^*)^{-1}\left[[\mathbf{x}_1^N + \mathbf{v}_1^N]_n,\dots, [\mathbf{x}_M^N + \mathbf{v}_M^N]_n\right]^\top, \, \forall n \in [N],
	\end{aligned}
	\end{equation}
	where $\{\mathbf{v}_m^N \sim \mathcal{N}(\mathbf{0},q_m^*\mathbf{I}_N)\}_{m=1}^M$ are mutually independent and are independent of $\{\mathbf{x}_m^N\}_{m=1}^M$. Note that the vector $\mathbf{x}_m^N + \mathbf{v}_m^N$ in \eqref{calculate-decoder-output} approximates the codeword corresponding to $\mathbf{g}_m^N$, $\forall m \in [M]$, according to the properties of jointly typical sequences. Finally, $\hat{\mathbf{g}}^{N}$ is computed using \eqref{g-hat}.

	\subsection{Aggregation Distortion Comparison}
	\label{Distortion-Comparison}
	
	\begin{figure}[htbp]
		\vspace{-0.1cm}
		\setlength{\belowcaptionskip}{-0.4cm}
		\centering
		\includegraphics[width=0.7\linewidth]{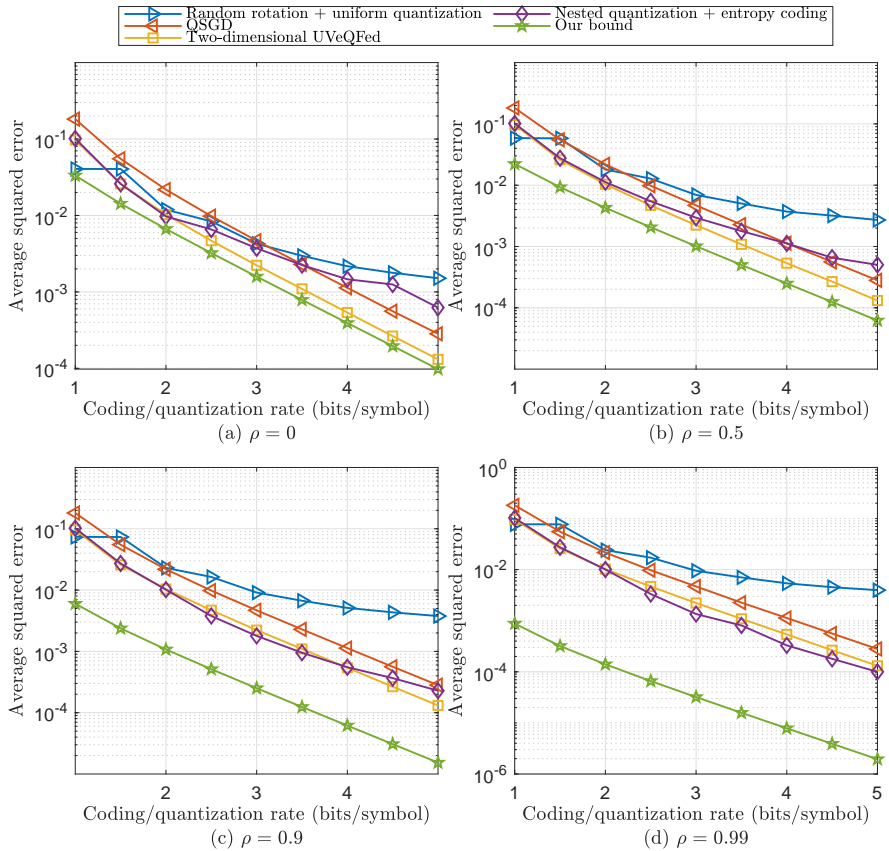}
		\caption{Average squared error versus coding/quantization rate (bits/symbol) with $\rho = 0$, $0.5$, $0.9$, and $0.99$.}
		\label{distortion-compare}
	\end{figure}
	
	
	In this subsection, the baseline schemes are compared (using synthetic data) with our bound in terms of the aggregation distortion with fixed source coding (quantization) rate $R$, i.e., the number of encoded (quantized) bits per source symbol. The baselines include
	QSGD \cite{Alistarh2017QSGD}, uniform quantization with random rotation \cite{Kone2016Federated}, two-dimensional
	UVeQFed \cite{Shlezinger2021UVeQFed}, nested quantization followed by entropy coding\footnote{Since the nested quantization scheme in \cite{Abdi2019Reducing} uses one-dimensional nested lattices to implement the lattice-based Wyner–Ziv coding scheme in \cite{Zamir2002Nested}, it cannot be guaranteed to decode successfully. We slightly improve the nested quantization scheme in \cite{Abdi2019Reducing} in the simulation: when the scheme cannot decode successfully, re-transmit.} \cite{Abdi2019Reducing}.
	

	The synthetic data is generated as follows. Let $\mathbf{s} \in \mathbb{R}^N$ and $\{\mathbf{w}_m \in \mathbb{R}^N\}_{m=1}^{M}$ be mutually independent random vectors with i.i.d. standard Gaussian elements. Given $\rho \in [0, 1]$, the synthetic data to be encoded/quantized is then given by
	$\mathbf{y}_m = \sqrt{\rho} \mathbf{s} + \sqrt{1 - \rho} \mathbf{w}_m$, $\forall m \in [M]$. Note that $\mathbf{y}_m$ also has i.i.d. standard Gaussian elements, $\forall m \in [M]$, and the correlation coefficient of $[\mathbf{y}_{m_1}]_{n}$ and $[\mathbf{y}_{m_2}]_n$ is exactly $\rho$, $\forall n \in [N]$, $m_1 \neq m_2 \in [M]$. Thus we can adjust the correlation between $\{\mathbf{y}_m\}_{m=1}^M$ by adjusting $\rho$. The target is to recover the aggregated vector $\mathbf{y} \triangleq \sum_{m=1}^{M}\mathbf{y}_m/M$. 
	
	As discussed in Section \ref{modified-Berger-Tung-coding}, our MBTC-based achievability scheme allows separate encoding on each $\mathbf{y}_m$ and joint decoding to directly obtain $\hat{\mathbf{y}}$ (an estimation of $\mathbf{y}$). For the baselines, we separately quantize $\mathbf{y}_m$ into $\hat{\mathbf{y}}_m$ using the corresponding quantization schemes with rate-determined quantization resolutions, and then obtain $\hat{\mathbf{y}}$ by 
	\begin{equation} \label{y_hat}
	\hat{\mathbf{y}} = \sum_{m=1}^{M}\hat{\mathbf{y}}_m/M.
	\end{equation}
	The distortion is measured by the \emph{average squared error} $\|\mathbf{y} - \hat{\mathbf{y}}\|^2/N$. In the simulation, we set the number of samples $N = 2^{17}$ and the number of sources $M = 10$.
	
	
	In Fig. \ref{distortion-compare}, we plot the average squared error versus the coding (quantization) rate with different values of $\rho$. As expected, the baselines are far from our bound when the correlation coefficient $\rho$ is large, suggesting that the baselines have great potential for further improvement when the sources (local updates) are strongly correlated. Note that the baselines are consistently worse than our bound, even when $\rho=0$. This is due to the fact that our bound is essentially obtained by infinite-length vector quantization.
	
	
	\subsection{FL Performance Comparison}
	\label{FL-Performance-Comparison}
	

	
	In this subsection, we compare the baselines with our bound in terms of the convergence performance and the communication cost. To this aim, we test our achievability scheme and the baselines by training a convolutional neural network (CNN), illustrated by Fig. \ref{CNN-structure}, in an FL fashion on the MNIST and the Fashion-MNIST datasets. The baselines include BSC \cite{Sattler2019Sparse}, top-k sparsification with residual accumulation \cite{Yujun2018Deep}, DPCM followed by entropy coding\footnote{The scheme ``DPCM followed by entropy coding'' treats each local update as a sample sequence drawn from a source with memory and then encodes the sequence using first-order differential pulse code modulation (DPCM), where the predictor weights are calculated statistically.}, and the baselines considered in the previous subsection\footnote{We do not consider scheme ``Random rotation + uniform quantization'' for comparison since it is inapplicable for the settings of Fig. \ref{general-figs}.}. We uniformly allocate the data of the training dataset to the devices and ensure the data allocated to each device is identically distributed. Within each training round, each device uses its local data to train an epoch of five stochastic gradient descent (SGD) iterations with learning rate $\eta = 0.1$. The local update is obtained by calculating the difference between the model parameters before and after the local training. Besides, we adopt a linear aggregation target function $\kappa(x_1, \dots, x_M) = \sum_{m=1}^{M} K_m x_m/K$. The FL parameters are summarized in TABLE \ref{Federated_learning_paramaters}.
	
	\begin{table}[]
		\centering
		\caption{Federated learning paramaters.}
		\label{Federated_learning_paramaters}
		\begin{tabular}{lcc}
			\hline
			& \multicolumn{1}{l}{Convergence performance comparison} & \multicolumn{1}{l}{Communication cost comparison} \\ \hline
			Machine learning model  & CNN illustrated in Fig. \ref{CNN-structure}            & CNN illustrated in Fig. \ref{CNN-structure}       \\
			Number of devices ($M$) & $8$                                                    & $20$                                              \\
			Learning rate ($\eta$)  & $0.1$                                                  & $0.1$                                             \\
			Local update times      & $5$                                                    & $5$                                               \\
			Local batch size        & $1500$                                                 & $600$                                             \\
			Data distribution       & $iid$                                                  & $iid$                                             \\ \hline
		\end{tabular}
	\end{table}

	\begin{figure}[]
		\vspace{-0.1cm}
		\setlength{\belowcaptionskip}{-0.4cm}
		\centering
		\includegraphics[width=10cm]{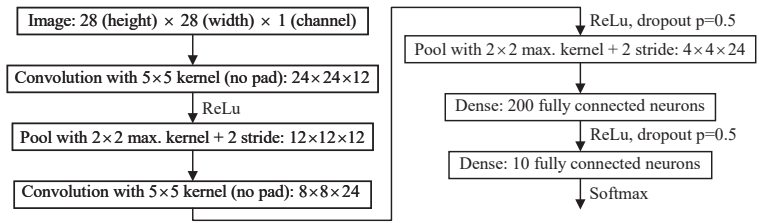}
		\caption{CNN structure.}
		\label{CNN-structure}
	\end{figure}

	\subsubsection{Convergence performance comparison}
	\label{Convergence-performance-compare}
	In this part, we compare the FL convergence performance of different schemes with identical coding (quantization) rate. We consider a wireless FL uplink for simulation. The system consists of one single-antenna PS and $M$ single-antenna devices, where the devices are randomly distributed inside a circle centered on the PS with radius $d$. For the uplink channels, we consider the Rayleigh fading model and the 5G Urban Macro (UMa) pathloss model given by European Telecommunications Standards Institute (ETSI) \cite{European20205G}. The channels are assumed to be unchanged during a communication round. For each device $m$, we assume that its uplink information rate can reach its channel capacity $C_m$, $m \in [M]$. Further, we assume that each device performs $\Upsilon$ times of channel realizations in each communication round. Then the maximum source coding (quantization) rate $r_m$ of device $m$ can be calculated by $r_m = \Upsilon C_m /N$, $\forall m \in [M]$, where $N = 86546$ according to the CNN structure. At each round, each scheme is adjusted to meet the maximum coding (quantization) rate constraints. Specifically, our achievability schemes can naturally meet these maximum rate constraints since they appear as constraints of problem (\ref{original-optimization-problem}).
	As for the baselines, we adjust their quantization resolutions to meet these constraints, and adopt \eqref{y_hat} for aggregation. In the simulation, we set $M = 8$, $d = 1.5$ km, 
	the power of additive white Gaussian noise (AWGN) as $-80$ dBm, the devices' transmitting power as $5$ dBm, the height of the PS antenna as $25$ m, the height of the device antennas as $1.5$ m, and the center frequency as $755$ MHz.
	
	
	
	\begin{figure}[htbp]
		\vspace{-0.1cm}
		\setlength{\belowcaptionskip}{-0.4cm}
		\centering
		\includegraphics[width=0.8\linewidth]{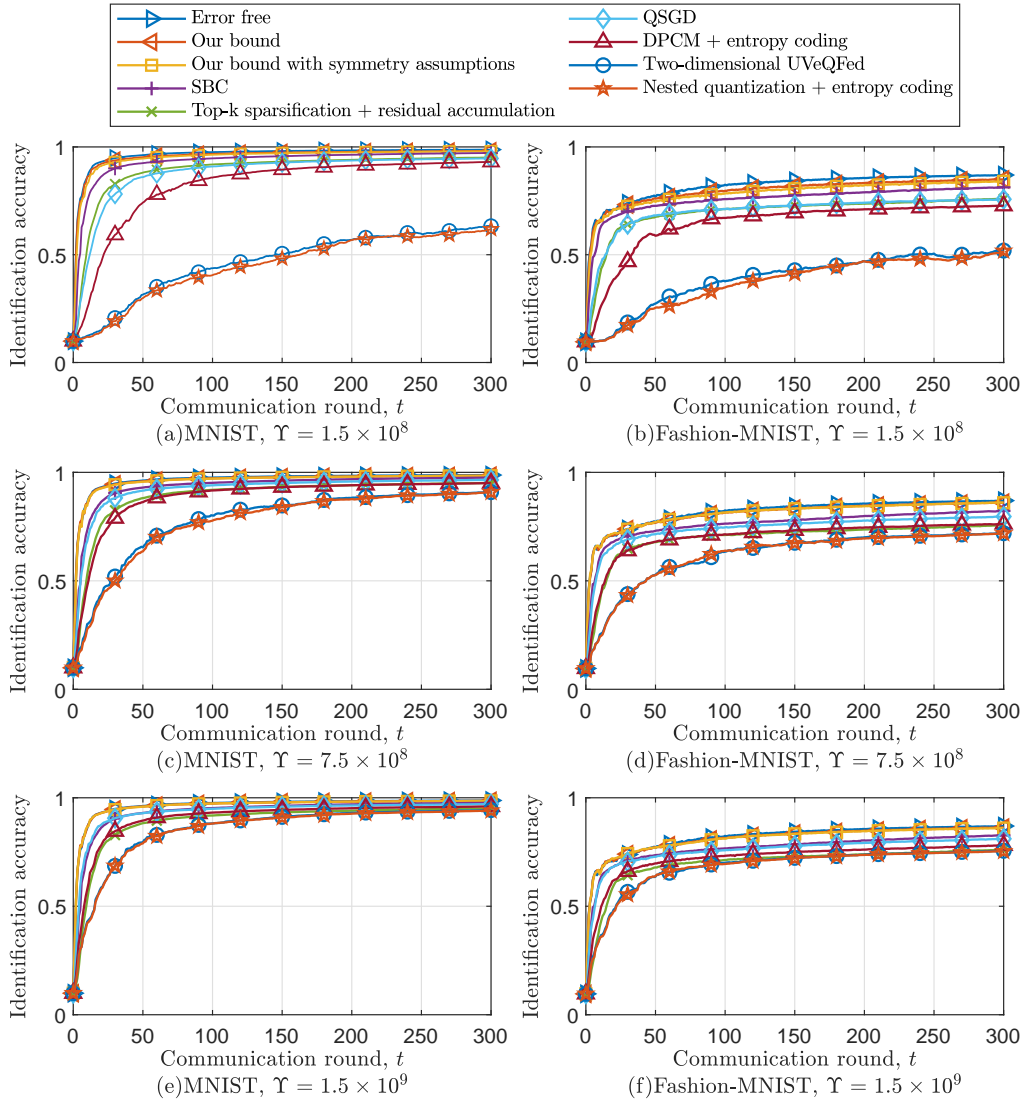}
		\caption{Classification accuracy versus communication round $t$ under MNIST and Fashion-MNIST datasets with different values of channel use times $\Upsilon$.}
		\label{general-figs}
	\end{figure}

	
	Fig. \ref{general-figs} plots the FL convergence performance of the considered schemes under MNIST and Fashion-MNIST datasets with the value of channel use times $\Upsilon$ adopted as $1.5 \times 10^8$, $7.5 \times 10^8$, and $1.5 \times 10^9$, respectively. It is worth emphasizing that when $\Upsilon$ takes $1.5 \times 10^8$, $7.5 \times 10^8$, and $1.5 \times 10^9$, each local update element is quantized on average to $0.603$ bits, $3.015$ bits, and $6.030$ bits, respectively. For the ``Our bound with symmetry assumptions'' case, we divide the eight devices into four groups in pairs such that the devices in the same group have relatively closer bit-constraints. For each group, we make the bit-constraints of both devices to be the same by tightening the looser one, which is always realizable.
	Besides, we assume that Assumption \ref{symmetry-assumption}-(\ref{symmetry-assumption1}) holds. 
	The above settings and assumption allow us to tune the MBTC parameters using Algorithm \ref{proposed-algorithm-for-grouping-case}. For the error free case, the local updates are aggregated without quantization. As shown in Fig. \ref{general-figs}, with the increase of the number of bits allowed to be transmitted in each round (realized by the increase of $\Upsilon$), the convergence rate and the final convergence point of every scheme increases. Further, we observe that our bounds are close to the error-free case, providing a promising convergence performance. Besides, we observe that the performance loss brought by devices grouping is negligible. This implies that, when the symmetry assumptions are satisfied, the grouping method is promising to reduce the complexity of coding parameters optimization with slight performance penalty.
	
	\begin{table*}[]
		\centering
		\caption{Rates ($\times 10^{-4}$ bits/symbol) required to reach a preset classification accuracy at the $100$-th communication round.}
		\label{rate-compare-table}
		\begin{tabular}{llrrrrrrr}
			\hline
			Dataset                              &  & \multicolumn{3}{c}{MNIST}                    & \multicolumn{1}{l}{} & \multicolumn{3}{c}{Fashion-MNIST}           \\ \cline{3-5} \cline{7-9} 
			Classification accuracy              &  & $90\%$        & $93\%$        & $95\%$       &                      & $72\%$        & $75\%$       & $78\%$       \\ \hline
			$2$-D UVeQFed                        &  & $462$         & $982$         & $1790$       &                      & $763$         & $1460$       & $3390$       \\
			Nested quantization $+$ entropy coding &  & $497$         & $1040$        & $2070$       &                      & $786$         & $1620$       & $3470$       \\
			QSGD                                 &  & $57.8$        & $196$         & $705$        &                      & $104$         & $462$        & $2310$       \\
			\textbf{Our bound}                        &  & \textbf{6.93} & \textbf{34.7} & \textbf{173} &                      & \textbf{23.1} & \textbf{127} & \textbf{832} \\ \hline
		\end{tabular}
	\end{table*}
	
	\subsubsection{Communication cost comparison}
	\label{Rate-compare}
	In this part, we compare the minimum source coding (quantization) rates required to achieve given convergence performances. We consider a fully symmetric FL system for simulation. Specifically, we assume that Assumptions \ref{symmetry-assumption}-(\ref{symmetry-assumption1}) and -(\ref{symmetry-assumption2}) hold and consider all the devices share the same source coding (quantization) rate. We set the number of devices $M = 20$ and the source coding (quantization) rate $R = r^{tot} / N$, where the number of symbols $N = 86546$ according to the CNN structure. In the simulation, we increase $r^{tot}$ initialized to a sufficiently small positive number with $100$-bit increment repeatedly until the classification accuracy at the $100$-th communication round is greater than a preset value. TABLE \ref{rate-compare-table} records this first-arrival rate for different schemes under different preset accuracies and datasets. We see that, to achieve a certain accuracy, the rates required by the baselines are all much greater than that needed by our achievability scheme (which leads to our bound).

		We now summarize the qualitative analysis of the numerical results in this section. Our bound is a ``good'' achievable bound mainly due to the following three properties of our proposed achievability scheme: 1) exploit the correlation between local updates for compression; 2) perform joint decoding at PS to directly reconstruct the global update, i.e., a linear combination of the local updates; 3) use infinite-length vector quantization. Due to 1), the baselines in Fig. 3 deviate farther and farther away from our bound as the data correlation increases; due to 3), the baselines in Fig. 3 are consistently worse than our bound, even if the data are completely uncorrelated. Properties 1) and 2) allow the re-allocation of the communication demands among devices by utilizing the correlation between local updates, thereby reducing the number of bits that need to be transmitted by deep-fading devices, explaining the ``good'' performance of our bound in Fig. 5. The results in TABLE II are mainly due to properties 1) and 3) since, under the considered fully symmetric simulation setup, the baseline's aggregation scheme (i.e., first decoding separately and then arithmetically averaging) is efficient enough.

	\section{Conclusion}
	\label{Conclusion}
	
	
In this paper, we studied the FL uplink from an information-theoretic perspective. We introduced a general performance analysis framework for model aggregation. We then characterized the performance limits of model aggregation in the form of an inner region of the rate-distortion region. Further, we developed two algorithms to search for the minimum aggregation distortion in the derived inner region for general and symmetric FL systems, respectively. Numerical results demonstrated that the baseline model aggregation schemes still have great potential for further improvement in the considered scenarios.
	

	\appendices
	
	\section{Proof of Proposition \ref{empirically-Gaussian}} \label{proof-empirically-Gaussian}
	
	Denote $\mathbf{z}_k^N \triangleq \mathbf{A}^{N} \mathbf{p}_k^N$, $\forall k \in [K]$. Since (i) $\mathbf{A}^{N}$ is Haar distributed, (ii) $\{\mathbf{p}_k^N\}_{k=2}^K$ are isotropically distributed, and (iii) $\{\mathbf{p}_k^N\}_{k=1}^K$ are mutually independent, we have $p(\mathbf{z}_{k}^N \mid \mathcal{Z}_k) = p(\mathbf{z}_{k}^N)$, $\forall \mathcal{Z}_k \subseteq \{\mathbf{z}_{i}^N\}_{i=1, i \neq k}^K$, $\mathcal{Z}_k \neq \O$, $k \in [K]$,
	i.e., $\{\mathbf{z}_k^N\}_{k=1}^K$ are mutually independent. Since $\mathbf{A}^{N}$ is Haar distributed, $\mathbf{z}_1^N$ is isotropically distributed; since orthogonal transformation does not change the distribution of an isotropically distributed vector, $\{\mathbf{z}_k^N\}_{k=2}^K$ are also isotropically distributed. Thus, for any $N\in\mathbb{Z}_+$, $\{\mathbf{z}_k^N / \|\mathbf{z}_k^N\|_2\}_{k=1}^K$ can be generated by a group of mutually independent Gaussian vectors $\{\mathbf{q}_k^N\sim \mathcal{N}(\mathbf{0},\mathbf{I}_N)\}_{k=1}^K$ through
	\begin{equation}
	\frac{\mathbf{z}_k^N}{\|\mathbf{z}_k^N\|_2} = \frac{\mathbf{q}_k^{N}}{\|\mathbf{q}_k^{N}\|_2}, \, \forall k \in [K].
	\end{equation}
	Then,
	\begin{equation}
	\label{01}
	\frac{1}{N}\left\|\mathbf{z}_k^N- \tau_k\mathbf{q}_k^{N}\right\|^2_2 = \frac{1}{N}\left\|\mathbf{z}_k^N\right\|_2^2 -\frac{2\tau_k}{N}\left\|\mathbf{z}_k^N\right\|_2 \cdot\left\|\mathbf{q}_k^{N}\right\|_2+\frac{\tau_k^2}{N}\left\|\mathbf{q}_k^{N}\right\|_2^2.
	\end{equation}
	By assumption,
	\begin{equation}
	\label{02}
	\lim_{N\rightarrow \infty} \frac{1}{N} \|\mathbf{z}_k^N\|^2_2 = \lim_{N\rightarrow \infty}  \frac{1}{N} \|\mathbf{p}_k^{N}\|^2_2 \overset{\text{a.s.}}{=} \tau_k^2,
	\end{equation}
	implying
	\begin{equation}
	\label{03}
	\lim_{N\rightarrow \infty}\frac{1}{\sqrt{N}}\left\|\mathbf{z}_k^N\right\|_2 \overset{\text{a.s.}}{=} \tau_k\ \textrm{and}\ \lim_{N\rightarrow \infty}  \frac{1}{N} \mathbb{E}[\|\mathbf{p}_k^{N}\|^2_2] = \tau_k^2.
	\end{equation}
	Since $\mathbf{q}_k^{N}\sim \mathcal{N}(\mathbf{0}, \mathbf{I}_N)$, according to the strong law of large numbers,
	\begin{equation}
	\label{04}
	\lim_{N\rightarrow \infty} \frac{1}{N} \|\mathbf{q}_k^N\|^2_2 \overset{\text{a.s.}}{=} 1.
	\end{equation}
	Note that $\|\mathbf{z}_k^N\|_2$ is a random variable independent of $\mathbf{q}_k^N$. Combining \eqref{01}-\eqref{04}, we obtain
	\begin{equation}
	\lim_{N\rightarrow \infty}\frac{1}{N}\left\|\mathbf{z}_k^N- \tau_k\mathbf{q}_k^{N}\right\|^2_2 \overset{\text{a.s.}}{=} 0.
	\end{equation}
	
	By the definitions,
	\begin{equation}
	\label{07}
	\begin{aligned}
	\sigma_{m_1, m_2}^2 = \lim_{N\rightarrow \infty} \frac{1}{N} \sum_{k=1}^K e_{m_1,k}e_{m_2,k} \mathbb{E}[\|\mathbf{p}_k^N\|_2^2]
	+ \lim_{N\rightarrow \infty} \frac{1}{N} \underset{k_1\neq k_2}{\sum_{k_1=1}^K \sum_{k_2=1}^K} e_{m_1,k_1} e_{m_2,k_2} \mathbb{E}[(\mathbf{p}_{k_1}^N)^\top\mathbf{p}_{k_2}^N].
	\end{aligned}
	\end{equation}
	By assumption, when $k_1 \!\neq\! k_2 \!\in\! [K]$, $\mathbf{p}_{k_1}^N$ and $\mathbf{p}_{k_2}^N$ are mutually independent and at least one of them is isotropically distributed. Without loss of generality, assume $\mathbf{p}_{k_2}^N$ is isotropically distributed. Then we have $\mathbb{E}[(\mathbf{p}_{k_1}^N)\!^\top\mathbf{p}_{k_2}^N \!\mid\! \mathbf{p}_{k_1}^N] = 0$, leading to $\mathbb{E}[(\mathbf{p}_{k_1}^N)\!^\top\mathbf{p}_{k_2}^N] = 0$. Combining with (\ref{03}) and (\ref{07}),
	\begin{equation}
	\sigma_{m_1,m_2}^2 = \sum_{k=1}^K e_{m_1,k}e_{m_2,k} \tau_k^2, \, \forall m_1,m_2 \in [M].
	\end{equation}
	
	Now, let $\tilde{\mathbf{x}}_m^N\triangleq \sum_{k=1}^K e_{m, k} \tau_k \mathbf{q}_k^N$, $\forall m \in [M]$. By the independence of $\{\mathbf{q}_k^N\}_{k=1}^K$, we have $\tilde{\mathbf{x}}_m^N \sim \mathcal{N}\left(\mathbf{0}, \sigma_{m,m}^2\mathbf{I}_N\right)$. Note that for any $n\in [N]$, the independent Gaussian variables $\{[\mathbf{q}^N_k]_n\}^K_{k=1}$ can be viewed as (a special case of) jointly Gaussian variables, and linear combination preserves the joint Gaussianity. Thus the entries $[\tilde{\mathbf{x}}_1^N]_n, [\tilde{\mathbf{x}}_2^N]_n, \dots, [\tilde{\mathbf{x}}_M^N]_n$ are jointly Gaussian with $\mathbb{E}\left[[\tilde{\mathbf{x}}_{m_1}^N]_n]\cdot[\tilde{\mathbf{x}}_{m_2}^N]_n]\right] = \sigma_{m_1,m_2}^2$, $\forall m_1,m_2 \in [M]$, $n\in[N]$. Moreover,
	\begin{align}
	\frac{1}{N}\left\| \mathbf{A}^N\tilde{\mathbf{g}}_m^N - \tilde{\mathbf{x}}_m^N \right\|^2_2=\frac{1}{ N}\left\|\sum_{k=1}^K e_{m, k} (\mathbf{z}_k^N - \tau_k \mathbf{q}_k^N) \right\|^2_2 \leq K\sum_{k=1}^K  e_{m,k}^2 \left(\frac{1}{N} \left\|\mathbf{z}_k^N- \tau_k\mathbf{q}_k^N\right\|^2_2\right) \rightarrow 0
	\end{align}
	almost surely as $N \rightarrow \infty$. We complete the proof.

	\section{Proof of Proposition \ref{FL-convergence}} \label{proof-FL-convergence}
	
	To standardize the notations, we rewrite the lemma used in the proof as follows.
	
	\begin{lemma}[{\cite[Lemma 2.1]{Friedlander2012Hybrid}}] \label{error-influence}
		Under Assumptions \ref{ass::convergence_condition}, set the learning rate $\eta = 1/\Omega$, we have
		\begin{equation}\label{error-influence-inequality}
		L(\boldsymbol{\theta}^{(t+1)}) - L(\boldsymbol{\theta}^{\star}) \leq \left(1 - \frac{\omega}{\Omega}\right)\left[L(\boldsymbol{\theta}^{(t)}) - L(\boldsymbol{\theta}^{\star})\right] + \frac{1}{2\Omega} \|\mathbf{e}^{(t)}\|^2,\forall t \geq 0,
		\end{equation}
		where the gradient error $\mathbf{e}^{(t)} \triangleq \hat{\mathbf{g}}^{(t)} - \mathbf{g}^{(t)}$.
	\end{lemma}
	From the recursion in Lemma \ref{error-influence}, we have
	\begin{equation}\label{recursion-result}
	\begin{aligned}
	L(\boldsymbol{\theta}^{(T+1)}) - L(\boldsymbol{\theta}^{\star}) \leq \left(1 - \frac{\omega}{\Omega}\right)^{T+1}\left[L(\boldsymbol{\theta}^{(0)}) - L(\boldsymbol{\theta}^{\star})\right] + \frac{1}{2\Omega} \sum_{t=0}^T \left(1 - \frac{\omega}{\Omega}\right)^{T-t} \|\mathbf{e}^{(t)}\|^2.
	\end{aligned}
	\end{equation}
	Recall that $\hat{\mathbf{g}}^{(t)}$ is the estimation of $\mathbf{g}^{(t)}$ given by the MBTC-based uplink scheme. Thus, by definitions,
	\begin{equation}
	\frac{1}{N}\|\mathbf{e}^{(t)}\|^2 = \frac{1}{N}\|\hat{\mathbf{g}}^{(t)} - \mathbf{g}^{(t)}\|^2 \leq D^{(t)},
	\end{equation}
	as $N \to \infty$. Thus, we approximately have $\|\mathbf{e}^{(t)}\|^2 \leq N D^{(t)}$.
	Combining with (\ref{recursion-result}), we obtain (\ref{FL-convergence-equation}).


	\section{Proof of Lemma \ref{mutual-information-upper-bound}} \label{proof-mutual-information-upper-bound}
	
	Let $r(\mathbf{u}^{\mathcal{S}} \mid \mathbf{u}^{\mathcal{S}^c})$ be any conditional probability density function. For any nonempty set $\mathcal{S} \subset [M]$,
	\begin{equation}\label{mutual-information-upper-bound-derivation-process-1}
	\begin{aligned} 
	I\left(\mathbf{x}^{\mathcal{S}};\, \mathbf{u}^\mathcal{S} \mid \mathbf{u}^{\mathcal{S}^c}\right)
	=&\int p(\mathbf{x}^{\mathcal{S}}, \mathbf{u}^\mathcal{S}, \mathbf{u}^{\mathcal{S}^c}) \log\frac{p(\mathbf{u}^\mathcal{S} \mid \mathbf{x}^{\mathcal{S}}, \mathbf{u}^{\mathcal{S}^c})}{p(\mathbf{u}^{\mathcal{S}} \mid \mathbf{u}^{\mathcal{S}^c})} \,
	\mathrm{d} \mathbf{x}^{\mathcal{S}}
	\mathrm{d}\mathbf{u}^{\mathcal{S}}
	\mathrm{d}\mathbf{u}^{\mathcal{S}^c} \\
	\stackrel{(a)}{=}& \int p( \mathbf{u}^\mathcal{S}, \mathbf{u}^{\mathcal{S}^c}) \log\frac{1}{p(\mathbf{u}^{\mathcal{S}} \mid \mathbf{u}^{\mathcal{S}^c})} \,
	\mathrm{d}\mathbf{u}^{\mathcal{S}}
	\mathrm{d}\mathbf{u}^{\mathcal{S}^c} 
	- h(\mathbf{u}^{\mathcal{S}}\mid \mathbf{x}^{\mathcal{S}})\\
	\stackrel{(b)}{\leq}& \int p( \mathbf{u}^\mathcal{S}, \mathbf{u}^{\mathcal{S}^c}) \log\frac{1}{r(\mathbf{u}^{\mathcal{S}} \mid \mathbf{u}^{\mathcal{S}^c})} \,
	\mathrm{d}\mathbf{u}^{\mathcal{S}}
	\mathrm{d}\mathbf{u}^{\mathcal{S}^c} 
	- h(\mathbf{u}^{\mathcal{S}}\mid \mathbf{x}^{\mathcal{S}}),
	\end{aligned}
	\end{equation}
	where $h(\mathbf{u}^{\mathcal{S}}\mid \mathbf{x}^{\mathcal{S}})$ denotes the conditional differential entropy of $\mathbf{u}^{\mathcal{S}}$ given $\mathbf{x}^{\mathcal{S}}$. Note that step ($a$) follows from the fact that $\mathbf{u}^\mathcal{S} \leftrightarrow \mathbf{x}^{\mathcal{S}}  \leftrightarrow \mathbf{u}^{\mathcal{S}^c}$, implying $p(\mathbf{u}^\mathcal{S} \mid \mathbf{x}^{\mathcal{S}}, \mathbf{u}^{\mathcal{S}^c}) = p(\mathbf{u}^\mathcal{S} \mid \mathbf{x}^{\mathcal{S}})$, and step ($b$) is obtained by the following result:
	\begin{equation}\label{KL-divergence-is-no-less-than-zero}
	\begin{aligned} 
	&\int p( \mathbf{u}^\mathcal{S}, \mathbf{u}^{\mathcal{S}^c}) \log\frac{1}{r(\mathbf{u}^{\mathcal{S}} \mid \mathbf{u}^{\mathcal{S}^c})} \,
	\mathrm{d}\mathbf{u}^{\mathcal{S}}
	\mathrm{d}\mathbf{u}^{\mathcal{S}^c}
	-
	\int p( \mathbf{u}^\mathcal{S}, \mathbf{u}^{\mathcal{S}^c}) \log\frac{1}{p(\mathbf{u}^{\mathcal{S}} \mid \mathbf{u}^{\mathcal{S}^c})} \,
	\mathrm{d}\mathbf{u}^{\mathcal{S}}
	\mathrm{d}\mathbf{u}^{\mathcal{S}^c}
	\\
	=& \int p(\mathbf{u}^{\mathcal{S}^c}) \int p(\mathbf{u}^{\mathcal{S}} \mid \mathbf{u}^{\mathcal{S}^c})
	\log\frac{p(\mathbf{u}^{\mathcal{S}} \mid \mathbf{u}^{\mathcal{S}^c})}{r(\mathbf{u}^{\mathcal{S}} \mid \mathbf{u}^{\mathcal{S}^c})} \,
	\mathrm{d}\mathbf{u}^{\mathcal{S}}
	\mathrm{d}\mathbf{u}^{\mathcal{S}^c} 
	\\
	=& \int p(\mathbf{u}^{\mathcal{S}^c}) \mathrm{KL}\left( p(\mathbf{u}^{\mathcal{S}} \mid \mathbf{u}^{\mathcal{S}^c})
	\| r(\mathbf{u}^{\mathcal{S}} \mid \mathbf{u}^{\mathcal{S}^c})\right) \,
	\mathrm{d}\mathbf{u}^{\mathcal{S}^c} \geq 0.
	\end{aligned}
	\end{equation}
	According to (\ref{KL-divergence-is-no-less-than-zero}), the equality in (\ref{mutual-information-upper-bound-derivation-process-1}) holds when $r(\mathbf{u}^{\mathcal{S}} \mid \mathbf{u}^{\mathcal{S}^c}) = p(\mathbf{u}^{\mathcal{S}} \mid \mathbf{u}^{\mathcal{S}^c})$.
	
	Recall that $\mathbf{x} \sim \mathcal{N}(\mathbf{0}, \boldsymbol{\Sigma}_{X})$, $\mathbf{v} \sim \mathcal{N}(\mathbf{0}, \boldsymbol{\Sigma}_{V})$ independently from $\mathbf{x}$, and $\mathbf{u} = \mathbf{x} + \mathbf{v}$. Then we have $p(\mathbf{u}) = \mathcal{N}(\mathbf{u}; \mathbf{0}, \boldsymbol{\Sigma}_{\mathrm{U}})$, $p(\mathbf{u}^{\mathcal{S}} \mid \mathbf{u}^{\mathcal{S}^c}) = \mathcal{N}(\mathbf{u}^{\mathcal{S}}; \boldsymbol{\Sigma}_{\mathrm{U}}^{\mathcal{S}, \mathcal{S}^c}(\boldsymbol{\Sigma}_{\mathrm{U}}^{\mathcal{S}^c})^{-1}\mathbf{u}^{\mathcal{S}^c},
	\boldsymbol{\Sigma}_{\mathrm{U}}^{\mathcal{S}} - \boldsymbol{\Sigma}_{\mathrm{U}}^{\mathcal{S}, \mathcal{S}^c}
	(\boldsymbol{\Sigma}_{\mathrm{U}}^{\mathcal{S}^c})^{-1}
	\boldsymbol{\Sigma}_{\mathrm{U}}^{\mathcal{S}^c, \mathcal{S}})$, and 
	\begin{equation}\label{p-u-S-x-S}
	p(\mathbf{u}^{\mathcal{S}} \mid \mathbf{x}^{\mathcal{S}}) 
	= \mathcal{N}(\mathbf{u}^{\mathcal{S}}; \mathbf{x}^{\mathcal{S}}, \boldsymbol{\Sigma}_{V}^{\mathcal{S}}),
	\end{equation}
	where $\boldsymbol{\Sigma}_{\mathrm{U}} =  \boldsymbol{\Sigma}_{X} + \boldsymbol{\Sigma}_{V}$. Note that \eqref{mutual-information-upper-bound-derivation-process-1} holds even when we restrict $r(\mathbf{u}^{\mathcal{S}} \mid \mathbf{u}^{\mathcal{S}^c})$ to be
	\begin{equation} \label{r}
	\setlength{\abovedisplayskip}{3pt}
	r(\mathbf{u}^{\mathcal{S}} \mid \mathbf{u}^{\mathcal{S}^c}) = \mathcal{N}(\mathbf{u}^{\mathcal{S}}; \mathbf{E}\mathbf{u}^{\mathcal{S}^c}, \mathbf{F}),
	\end{equation}
	where auxiliary matrices $\mathbf{E} \in \mathbb{R}^{|\mathcal{S}| \times |\mathcal{S}^c|}$ and $\mathbf{F} \in \mathbb{R}^{|\mathcal{S}| \times |\mathcal{S}|}$ with $\mathbf{F} \succ 0$. In this case, the equality in \eqref{mutual-information-upper-bound-derivation-process-1} holds when $\mathbf{E} = \boldsymbol{\Sigma}_{\mathrm{U}}^{\mathcal{S}, \mathcal{S}^c}(\boldsymbol{\Sigma}_{\mathrm{U}}^{\mathcal{S}^c})^{-1}$ and $\mathbf{F} = \boldsymbol{\Sigma}_{\mathrm{U}}^{\mathcal{S}} - \boldsymbol{\Sigma}_{\mathrm{U}}^{\mathcal{S}, \mathcal{S}^c}
	(\boldsymbol{\Sigma}_{\mathrm{U}}^{\mathcal{S}^c})^{-1}
	\boldsymbol{\Sigma}_{\mathrm{U}}^{\mathcal{S}^c, \mathcal{S}}$. By substituting \eqref{p-u-S-x-S} and (\ref{r}) into the right-hand side of (\ref{mutual-information-upper-bound-derivation-process-1}), we have
	\begin{equation}\label{mutual-information-upper-bound-derivation-process-2}
	\begin{aligned} 
	& \int p( \mathbf{u}^\mathcal{S}, \mathbf{u}^{\mathcal{S}^c}) \log\frac{1}{r(\mathbf{u}^{\mathcal{S}} \mid \mathbf{u}^{\mathcal{S}^c})} \,
	\mathrm{d}\mathbf{u}^{\mathcal{S}}
	\mathrm{d}\mathbf{u}^{\mathcal{S}^c}\\
	=& \frac{1}{2}\log\left((2 \pi)^{|\mathcal{S}|}\mathrm{det}(\mathbf{F})\right)
	+ \frac{\log(e)}{2}  \mathbb{E}_{\mathbf{u}^{\mathcal{S}}, \mathbf{u}^{\mathcal{S}^c}}\left[\left(\mathbf{u}^{\mathcal{S}} - \mathbf{E}\mathbf{u}^{\mathcal{S}^c}\right)^{\top} \mathbf{F}^{-1} \left(\mathbf{u}^{\mathcal{S}} - \mathbf{E}\mathbf{u}^{\mathcal{S}^c}\right)\right] 
	\end{aligned}
	\end{equation}
	where
	\begin{equation}\label{mutual-information-upper-bound-derivation-process-2-plus}
	\begin{aligned} 
	&\mathbb{E}_{\mathbf{u}^{\mathcal{S}}, \mathbf{u}^{\mathcal{S}^c}}\!\left[\left(\mathbf{u}^{\mathcal{S}} \!-\! \mathbf{E}\mathbf{u}^{\mathcal{S}^c}\right)^{\top} \mathbf{F}^{-1} \left(\mathbf{u}^{\mathcal{S}} \!-\! \mathbf{E}\mathbf{u}^{\mathcal{S}^c}\right)\right]
	=\mathrm{tr}\left\{\mathbf{F}^{-1}\mathbb{E}_{\mathbf{u}^{\mathcal{S}}, \mathbf{u}^{\mathcal{S}^c}}\!\left[\left(\mathbf{u}^{\mathcal{S}} \!-\! \mathbf{E}\mathbf{u}^{\mathcal{S}^c}\right)\left(\mathbf{u}^{\mathcal{S}} \!-\! \mathbf{E}\mathbf{u}^{\mathcal{S}^c}\right)^{\top}\right]\right\}\\
	\stackrel{(a)}{=}&
	\mathrm{tr}\left\{\mathbf{F}^{-1} \boldsymbol{\Sigma}_{V}^{\mathcal{S}}\right\}
	\!+\!
	\mathrm{tr}\left\{\mathbf{E}^{\top}\mathbf{F}^{-1}\mathbf{E} \boldsymbol{\Sigma}_{V}^{\mathcal{S}^c}\right\} 
	\!+\!
	\mathrm{tr}\left\{\mathbf{F}^{-1}\left( \boldsymbol{\Sigma}_{X}^{\mathcal{S}}
	\!+\!\mathbf{E}\boldsymbol{\Sigma}_{X}^{\mathcal{S}^c}\mathbf{E}^{\top}
	\!-\!\mathbf{E}\boldsymbol{\Sigma}_{X}^{\mathcal{S}^c, \mathcal{S}}
	\!-\!\boldsymbol{\Sigma}_{X}^{\mathcal{S}, \mathcal{S}^c}\mathbf{E}^{\top}
	\right)\right\},   
	\end{aligned}
	\end{equation}
	with step ($a$) follows from $\boldsymbol{\Sigma}_{U} = \boldsymbol{\Sigma}_{X} + \boldsymbol{\Sigma}_{V}$ with $\boldsymbol{\Sigma}_{V}$ be a diagonal matrix (implying $\boldsymbol{\Sigma}_{U}^{\mathcal{S}} = \boldsymbol{\Sigma}_{X}^{\mathcal{S}} + \boldsymbol{\Sigma}_{V}^{\mathcal{S}}$, $\boldsymbol{\Sigma}_{U}^{\mathcal{S}^c} = \boldsymbol{\Sigma}_{X}^{\mathcal{S}^c} + \boldsymbol{\Sigma}_{V}^{\mathcal{S}^c}$, $\boldsymbol{\Sigma}_{U}^{\mathcal{S}, \mathcal{S}^c} = \boldsymbol{\Sigma}_{X}^{\mathcal{S}, \mathcal{S}^c}$, and $\boldsymbol{\Sigma}_{U}^{\mathcal{S}^c, \mathcal{S}} = \boldsymbol{\Sigma}_{X}^{\mathcal{S}^c, \mathcal{S}}$).
	Similarly, we have
	\begin{equation}\label{mutual-information-upper-bound-derivation-process-3}
	\begin{aligned} 
	h(\mathbf{u}^{\mathcal{S}}\mid \mathbf{x}^{\mathcal{S}})
	=& \frac{1}{2}\log\left((2 \pi)^{|\mathcal{S}|}\mathrm{det}(\boldsymbol{\Sigma}_{V}^{\mathcal{S}})\right) + \frac{\log(e)}{2}  \mathbb{E}_{\mathbf{x}^{\mathcal{S}}, \mathbf{u}^{\mathcal{S}}}\left[\left(\mathbf{u}^{\mathcal{S}} - \mathbf{x}^{\mathcal{S}}\right)^{\top} \left(\boldsymbol{\Sigma}_{V}^{\mathcal{S}}\right)^{-1} \left(\mathbf{u}^{\mathcal{S}} - \mathbf{x}^{\mathcal{S}}\right)\right]\\
	=& \frac{1}{2}\log\left((2 \pi)^{|\mathcal{S}|}\mathrm{det}(\boldsymbol{\Sigma}_{V}^{\mathcal{S}})\right)
	+ \frac{\log(e)}{2}  \mathbb{E}_{\mathbf{v}^{\mathcal{S}}}\left[ \mathrm{tr}\left\{ \left(\boldsymbol{\Sigma}_{V}^{\mathcal{S}}\right)^{-1} \mathbf{v}^{\mathcal{S}} \left(\mathbf{v}^{\mathcal{S}}\right)^{\top}\right\}\right]\\
	=& \frac{1}{2}\log\left((2 \pi e)^{|\mathcal{S}|}\mathrm{det}(\boldsymbol{\Sigma}_{V}^{\mathcal{S}})\right).
	\end{aligned}
	\end{equation}
	Combining (\ref{mutual-information-upper-bound-derivation-process-1}) and (\ref{mutual-information-upper-bound-derivation-process-2})-(\ref{mutual-information-upper-bound-derivation-process-3}), we finally obtain
	\begin{equation}\label{mutual-information-upper-bound-derivation-process-4}
	\setlength{\abovedisplayskip}{3pt}
	\begin{aligned} 
	I\left(\mathbf{x}^{\mathcal{S}};\, \mathbf{u}^\mathcal{S} \mid \mathbf{u}^{\mathcal{S}^c}\right)\leq
	\chi_{\mathcal{S}}(\mathbf{E}, \mathbf{F}, \boldsymbol{\Sigma}_{V}^{\mathcal{S}})
	+ \xi_{\mathcal{S}}(\mathbf{E}, \mathbf{F}),
	\end{aligned}
	\end{equation}
	where the equality holds when $\mathbf{E} = \boldsymbol{\Sigma}_{X}^{\mathcal{S}, \mathcal{S}^c}(\boldsymbol{\Sigma}_{X}^{\mathcal{S}^c} + \boldsymbol{\Sigma}_{V}^{\mathcal{S}^c})^{-1}$ and $\mathbf{F} = \boldsymbol{\Sigma}_{X}^{\mathcal{S}} + \boldsymbol{\Sigma}_{V}^{\mathcal{S}} - \boldsymbol{\Sigma}_{X}^{\mathcal{S}, \mathcal{S}^c}
	(\boldsymbol{\Sigma}_{X}^{\mathcal{S}^c} + \boldsymbol{\Sigma}_{V}^{\mathcal{S}^c})^{-1}
	\boldsymbol{\Sigma}_{X}^{\mathcal{S}^c, \mathcal{S}}$.
	
	In the same way, we obtain an upper bound of $I(\mathbf{x};\,\mathbf{u})$:
	\begin{equation}\label{upper-bound-of-I-x-u}
	\setlength{\abovedisplayskip}{3pt}
	\begin{aligned} 
	I\left(\mathbf{x};\, \mathbf{u}\right)
	\leq
	\chi_{[M]}(\mathbf{G}, \boldsymbol{\Sigma}_{V})
	+ \xi_{[M]}(\mathbf{G}),
	\end{aligned}
	\end{equation}
	where $\mathbf{G} \in \mathbb{R}^{M \times M}$ with $\mathbf{G} \succ 0$. The equality in \eqref{upper-bound-of-I-x-u} holds when $\mathbf{G} = \boldsymbol{\Sigma}_{X} + \boldsymbol{\Sigma}_{V}$.

	\section{Problem Transformation} \label{Problem-Form-Transformation}
	
	We first consider the objective of problem (\ref{original-optimization-problem-three}). According to Assumptions \ref{symmetry-assumption}-(\ref{symmetry-assumption1}) and -(\ref{symmetry-assumption2}), we have
	\begin{equation}
	\boldsymbol{\Sigma}_{X}^{\top} \mathbf{c} = ((M-1)\rho+1) \lambda \sigma_{X}^2 \cdot \mathbf{1}.
	\end{equation}
	Let $\boldsymbol{\Sigma}_{temp}=(\boldsymbol{\Sigma}_{V} + (1-\rho)\sigma_{X}^2 \mathbf{I})^{-1}$. According to the Sherman–Morrison formula, we have
	\begin{equation}
	(\boldsymbol{\Sigma}_{X} + \boldsymbol{\Sigma}_{V})^{-1} = \boldsymbol{\Sigma}_{temp} - \frac{\rho\sigma_{X}^2 \boldsymbol{\Sigma}_{temp}\mathbf{1}\mathbf{1}^{\top} \boldsymbol{\Sigma}_{temp}}{1+\rho\sigma_{X}^2 \mathbf{1}^{\top}\boldsymbol{\Sigma}_{temp}\mathbf{1}}.
	\end{equation}
	Together, the objective function of problem (\ref{original-optimization-problem-three}) can be rewritten as
	\begin{equation}
	\label{objective-new}
	\begin{aligned}
	\mathbf{c}^{\top} \boldsymbol{\Sigma}_{X}(\boldsymbol{\Sigma}_{X} + \boldsymbol{\Sigma}_{V})^{-1} \boldsymbol{\Sigma}_{X}^{\top} \mathbf{c} = \left[((M-1)\rho+1) \lambda \sigma_{X}^2\right]^2 \bigg/\left(\frac{1}{\mathbf{1}^{\top}\boldsymbol{\Sigma}_{temp}\mathbf{1}} + \rho \sigma_{X}^2\right).
	\end{aligned}
	\end{equation}
	To maximize (\ref{objective-new}) by tuning $\{q_{(j)}\}_{j=1}^J$, we only need to maximize $\mathbf{1}^{\top}\boldsymbol{\Sigma}_{temp}\mathbf{1}$, which equals to $\sum_{j=1}^J M_j / (q_{(j)} + (1-\rho) \sigma_{X}^2)$ according to (\ref{original-optimization-problem-three-c_d}). This gives the objective of problem (\ref{original-optimization-problem-four}).
	
	We then consider the constraints of problem (\ref{original-optimization-problem-three}). Since constraint (\ref{original-optimization-problem-three-c_d}) is simple variable substitutions, we focus on constraint (\ref{original-optimization-problem-three-c_b}). From (\ref{mutual-information-1}) and (\ref{mutual-information-2}), constraint (\ref{original-optimization-problem-three-c_b}) is given by
	\begin{equation}
	\label{App-D-5}
	\begin{aligned}
	&I\left(\mathbf{x}^{\mathcal{S}};\, \mathbf{u}^\mathcal{S} \mid \mathbf{u}^{\mathcal{S}^c}\right)
	=\frac{1}{2}\log\left(\frac{\det\left(\boldsymbol{\Sigma}_{X} + \boldsymbol{\Sigma}_{V}\right)}{\det\left(\boldsymbol{\Sigma}_{X}^{\mathcal{S}^c} + \boldsymbol{\Sigma}_{V}^{\mathcal{S}^c}\right) \det\left(\boldsymbol{\Sigma}_{V}^{\mathcal{S}}\right)}\right)\leq \sum_{m\in\mathcal{S}}r_m,
	\forall \mathcal{S} \subset [M],\, \mathcal{S} \neq \O,
	\end{aligned}
	\end{equation}
	\begin{equation}
	\label{App-D-6}
	I\left(\mathbf{x};\, \mathbf{u}\right) = \frac{1}{2}\log\left(\frac{\det\left(\boldsymbol{\Sigma}_{X} + \boldsymbol{\Sigma}_{V}\right)}{\det\left(\boldsymbol{\Sigma}_{V}\right)}\right)\leq \sum_{m\in[M]}r_m.
	\end{equation}
	Further transformation relies on the following key observation. Given nonempty sets $\mathcal{S}_1, \,\mathcal{S}_2 \subset [M]$ satisfy $|\mathcal{S}_1|=|\mathcal{S}_2|$. 
	If the elements in sets $\mathcal{S}_1$ and $\mathcal{S}_2$ are clustered in the same way (i.e., the number of elements assigned to each group is the same), the two constraints corresponding to these two sets have exactly the same form, i.e., they degenerate into one constraint. This inspires us to distinguish constraints by the number of devices selected (by a set $\mathcal{S}$) in each group. Specifically, let $\varsigma_j \in \{0\} \cup [M_j]$ denote the number of devices selected (by a set $\mathcal{S}$) in group $j$, $\forall j \in [J]$. Note that for a matrix $\mathbf{A} = b\cdot\mathbf{1}\mathbf{1}^{\top} + \mathrm{diag}([a_1, \dots, a_M])$, its determinant $\mathrm{det}(\mathbf{A}) = (1 + \sum_{m=1}^M$ $b / a_m) \prod_{m=1}^M a_m$. Thus for any set $\mathcal{S}$ satisfying $|\{s| s\in \mathrm{group}_j, \, s \in \mathcal{S} \}| = \varsigma_j$, $\forall j \in [J]$, we have
	\begin{equation}
	\label{App-D-1}
	\det\left(\boldsymbol{\Sigma}_{V}^{\mathcal{S}}\right) = \prod_{j=1}^J q_{(j)}^{\varsigma_j},
	\end{equation}
	\begin{equation}
	\label{App-D-2}
	\begin{aligned}
	\det\left(\boldsymbol{\Sigma}_{X}^{\mathcal{S}^c} + \boldsymbol{\Sigma}_{V}^{\mathcal{S}^c}\right)
	= \left(1 + \sum_{j=1}^J \frac{(M_j - \varsigma_j) \rho \sigma_{X}^2}{(1 - \rho)\sigma_{X}^2 + q_{(j)}} \right) \prod_{j=1}^J [(1-\rho) \sigma_{X}^2 + q_{(j)}]^{M_j - \varsigma_j}.
	\end{aligned}
	\end{equation}
	Similarly,
	\begin{equation}
	\label{App-D-3}
	\det\left(\boldsymbol{\Sigma}_{V}\right) = \prod_{j=1}^J q_{(j)}^{M_j},
	\end{equation}
	\begin{equation} 
	\label{App-D-4}
	\begin{aligned}
	\det\left(\boldsymbol{\Sigma}_{X} + \boldsymbol{\Sigma}_{V}\right)
	= \left(1 + \sum_{j=1}^J \frac{M_j \rho \sigma_{X}^2}{(1 - \rho)\sigma_{X}^2 + q_{(j)}} \right) \prod_{j=1}^J [(1-\rho) \sigma_{X}^2 + q_{(j)}]^{M_j}.
	\end{aligned}
	\end{equation}
	Substituting (\ref{App-D-1})-(\ref{App-D-4}) into (\ref{App-D-5}) and (\ref{App-D-6}), we obtain the left-hand side of (\ref{original-optimization-problem-four-constraints}). The right-hand side of (\ref{original-optimization-problem-four-constraints}) can be directly obtained according to Assumption \ref{symmetry-assumption}-(\ref{symmetry-assumption3}).

	\ifCLASSOPTIONcaptionsoff
	\newpage
	\fi

	
	
	%

	\bibliographystyle{IEEEtran}
	
	\bibliography{ref}

	


	%

	
	
	
	
	
	

\end{document}